\documentclass[11pt]{article}

\pdfoutput=1
\usepackage[margin=1in]{geometry}
\usepackage{amsmath,amsthm,amsfonts,amssymb}
\usepackage{graphicx}
\usepackage{microtype}
\usepackage{braket}

\usepackage[usenames,dvipsnames]{color}
\usepackage[colorlinks=true,citecolor=blue,linkcolor=blue]{hyperref}
\usepackage{placeins}
\usepackage[capitalize,nameinlink]{cleveref} 
\usepackage{enumitem}
\usepackage{thm-restate}
\usepackage[toc]{multitoc}

\newtheorem{theorem}{Theorem}
\newtheorem{lemma}[theorem]{Lemma}
\newtheorem{proposition}[theorem]{Proposition}

\theoremstyle{definition}

\newcommand{\ZZ}{{\mathbb{Z}}}
\newcommand{\CC}{{\mathbb{C}}}
\newcommand{\RR}{{\mathbb{R}}}

\newcommand{\id}{{\mathbf{1}}}
\newcommand{\rd}{{\mathrm{d}}}

\newcommand{\calO}{{{\mathcal O}}}
\newcommand{\tOmega}{{\widetilde{\Omega}}}
\newcommand{\tO}{{\widetilde{\calO}}}

\newcommand{\IS}{{\mathcal S}}
\newcommand{\dt}{{\rm d}t}
\newcommand{\odts}{{\calO}(\dt^2)}
\renewcommand{\i}{{\mathbf i}}

\newcommand{\norm}[1]{\left\| {#1}\right\|}
\newcommand{\trnorm}[1]{\left\| {#1}\right\|_{\mathrm{tr}}}
\newcommand{\abs}[1]{\left| {#1}\right|}

\DeclareMathOperator{\polylog}{polylog}
\DeclareMathOperator{\poly}{poly}
\DeclareMathOperator{\diam}{diam}
\DeclareMathOperator{\dist}{dist}
\DeclareMathOperator{\Tr}{Tr}

\renewcommand{\>}{\rangle}

\makeatletter
\newcommand{\para}{
	\@startsection{paragraph}{4}
	{\z@}{2ex \@plus 3.3ex \@minus .2ex}{-1em}
	{\normalfont\normalsize\bfseries}
}
\makeatother

\begin{document}

\title{\bfseries Quantum algorithm for simulating real\\ time evolution of lattice Hamiltonians}
\author{Jeongwan Haah\thanks{Microsoft Quantum and Microsoft Research, Redmond, Washington, USA} \quad Matthew B. Hastings$^{*,}$\thanks{Station Q, Microsoft Research, Santa Barbara, California, USA}\quad Robin Kothari$^*$\quad Guang Hao Low$^*$\\[.25em]
\texttt{\{jwhaah,mahastin,Robin.Kothari,GuangHao.Low\}@microsoft.com}}

\date{8 February 2020%
\footnote{A shorter version of this paper has appeared in 
\href{https://doi.org/10.1109/FOCS.2018.00041}{
{\it IEEE 59th Annual Symposium on Foundations of Computer Science} 
(FOCS), Paris, 2018, pp. 350-360}.
}
}

\maketitle

\begin{abstract}
    We study the problem of simulating the time evolution of a lattice Hamiltonian, where the qubits are laid out on a lattice and  the Hamiltonian only includes geometrically local interactions (i.e., a qubit may only interact with qubits in its vicinity). This class of Hamiltonians is very general and is believed to capture fundamental interactions of physics.
    
    Our algorithm simulates the time evolution of such a Hamiltonian on $n$ qubits for time $T$ up to error $\epsilon$ using $\calO(  nT \polylog (nT/\epsilon))$ gates with depth $\calO(T \polylog (nT/\epsilon))$. 
    Our algorithm is the first simulation algorithm that achieves gate cost quasilinear in $nT$ and polylogarithmic in $1/\epsilon$. 
    Our algorithm also readily generalizes to time-dependent Hamiltonians and yields an algorithm with similar gate count for any piecewise slowly varying time-dependent bounded local Hamiltonian. 
    
    We also prove a matching lower bound on the gate count of such a simulation, showing that any quantum algorithm that can simulate a piecewise constant bounded local Hamiltonian in one dimension to constant error requires $\tOmega(nT)$ gates in the worst case. The lower bound holds even if we only require the output state to be correct on local measurements. To our best knowledge, this is the first nontrivial lower bound on the gate complexity of the simulation problem.
    
    Our algorithm is based on a decomposition of the time-evolution unitary into a product of small unitaries using Lieb-Robinson bounds. In the appendix, we prove a Lieb-Robinson bound tailored to Hamiltonians with small commutators between local terms,
    giving zero Lieb-Robinson velocity in the limit of commuting Hamiltonians. This improves the performance of our algorithm when the Hamiltonian is close to commuting.
\end{abstract}

\section{Introduction}
    
\para{Background}

The problem of simulating the time evolution of a quantum system is perhaps the most important application of quantum computers. Indeed, this was the reason Feynman proposed quantum computing~\cite{Fey82}, and it remains an important practical application since a significant fraction of the world's supercomputing power is used to solve instances of this problem that arise in materials science, condensed matter physics, high energy physics, and chemistry~\cite{NERSC16}.
    
    All known classical algorithms (i.e., algorithms that run on traditional non-quantum computers) for this problem run in exponential time. On the other hand, from the early days of quantum computing~\cite{Fey82,Lloyd1996} it was known that quantum computers can solve this problem in polynomial time. More precisely, when formalized as a decision problem, the problem of simulating the time evolution of a quantum system is in the complexity class $\mathsf{BQP}$, the class of problems solved by a quantum computer to bounded error in polynomial time. Furthermore, the problem is complete for $\mathsf{BQP}$~\cite{Feynman1985,Chi04,Nagaj2009},
    which means we do not expect there to be efficient classical algorithms for the problem, since that would imply $\mathsf{BPP}=\mathsf{BQP}$, which in turn would imply polynomial-time algorithms for problems such as integer factorization and discrete log~\cite{Sho97}.

\para{Hamiltonian simulation problem} 

The Hamiltonian simulation problem is a standard formalization of the problem of simulating the time evolution\footnote{This is sometimes referred to as ``real time evolution'', to distinguish it from ``imaginary time evolution'' which we will not talk about in this paper.} of a quantum system. In this problem, we assume the quantum system whose time evolution we wish to simulate consists of $n$ qubits and we want to simulate its time evolution for time $T$, in the sense that we are provided with the initial state $|\psi(0)\>$ and we want to compute the state of the system at time $T$, $|\psi(T)\>$. The goal of an efficient simulation is to solve the problem in time polynomial in $n$ and $T$.
    
    The state of a system of $n$ qubits can be described by a complex vector of dimension $2^n$ of unit norm. Since we are studying quantum algorithms for the problem, we are given the input as an $n$-qubit quantum state, and have to output an $n$-qubit quantum state. The relation between the output state at time $T$ and the initial state at time $0$ is given by the Schr\"odinger equation
    \begin{align}
    \i\frac{\mathrm{d}}{\mathrm{d}t}|\psi(t)\>=H(t)|\psi(t)\>,
    \end{align}
    where the Hamiltonian $H$, a $2^n \times 2^n$ complex Hermitian matrix, has entries which may also be functions of time. The Hamiltonian captures the interaction between the constituents of the system and governs time dynamics. In the special case where the Hamiltonian is independent of time, the  Schr\"odinger equation can be solved to yield  $|\psi(T)\>=e^{-iHT}|\psi(0)\>$.
    
More formally, the input to the Hamiltonian simulation problem consists of a Hamiltonian $H$ 
(or $H(t)$ in the time-dependent case), a time $T$, and an error parameter $\epsilon$. 
The goal is to output a quantum circuit that approximates the unitary matrix 
that performs the time evolution above 
(e.g., for time-independent Hamiltonians, the quantum circuit should approximate the unitary $e^{-iHT}$). 
The notion of approximation used is the spectral norm distance 
between the ideal unitary $U$ and the one $V$ performed by the circuit.
This implies that the circuit as a quantum channel is close to the ideal one
in the completely bounded trace norm distance:
\begin{align*}
\sup_\rho \frac 1 2 \trnorm{ (V\otimes I) \rho (V\otimes I)^\dagger - (U\otimes I) \rho (U\otimes I)^\dagger }
\le \norm{V - U} \label{eq:cbnorm-opnorm}
\end{align*}
where $\trnorm{A} = \Tr \sqrt{A^\dagger A}$ is the trace norm, $\norm{ \cdot }$ is the spectral norm,
$I$ is the identity on an ancilla system, and $\rho$ ranges over all density matrices on the joint system.
The inequality is easily seen by using the triangle inequality of the trace norm 
and the inequality $\trnorm{AB} \le \trnorm{A}\norm{B}$ for any~$A,B$~\cite[Eq.~1.175]{WatrousBook}.
    
    The cost of a quantum circuit is measured by the number of gates used in the circuit, where the gates come from some fixed universal gate set. Note that it is important to describe how the Hamiltonian in the input is specified, since it is a matrix of size $2^n \times 2^n$. This will be made precise when talking about specific classes of Hamiltonians that we would like to simulate.
    
\para{Geometrically local Hamiltonian simulation}
    The most general class of Hamiltonians that is commonly studied in the literature is the class of sparse Hamiltonians~\cite{AT03,Chi04,BAC+07,BerryEtAl2014,TS,BCK15,QSP,LC16, LC17USA}. A Hamiltonian on $n$ qubits is sparse if it has only $\poly(n)$ nonzero entries in any row or column. For such Hamiltonians, we assume we have an efficient algorithm to compute the nonzero entries in each row or column, and the input Hamiltonian is specified by an oracle that can be queried for this information. In this model, recent quantum algorithms have achieved optimal complexity in terms of the queries made to the oracle~\cite{QSP}.
    
    A very important special case of this type of Hamiltonian is a ``local Hamiltonian.'' Confusingly, this term is used to describe two different kinds of Hamiltonians in the literature. 
    We distinguish these two definitions of ``local'' by referring to them as ``non-geometrically local'' and ``geometrically local'' in this introduction. A non-geometrically local Hamiltonian is a Hamiltonian $H(t)$ that can be written as a sum of polynomially many terms $H_j(t)$, each of which acts nontrivially on only $k$ qubits at a time (i.e., the matrix acts as identity on all other qubits). A geometrically local Hamiltonian is similar, except that each term $H_j(t)$ must act on $k$ adjacent qubits. Since we refer to ``adjacent'' qubits, the geometry of how the qubits are laid out in space must be specified. In this paper we will deal with qubits laid out in a $D$-dimensional lattice in Euclidean space. That is, qubits are located at points in $\mathbb{Z}^D$ and are adjacent if they are close in Euclidean distance.
In this paper, we consider $D$ as constant so asymptotic expressions will not show any dependence on $D$.
(However, we briefly discuss a strategy for large $D$ in \cref{app:extensions}.)

Geometrically local Hamiltonians are central objects in physics.
Indeed, fundamental forces of Nature (strong, weak, and electromagnetic interactions)
are modeled by geometrically local Hamiltonians on lattices~\cite{KogutSusskind1975}.
From a practical perspective, geometrically local Hamiltonians capture a large fraction of 
condensed matter systems we are interested in.\footnote{There are some physical situations where we do care about more general Hamiltonians. Even though the system we are given may be described by a geometrically local Hamiltonian, it is sometimes computationally advantageous to represent a given system with a non-geometrically local (or sparse) Hamiltonian.}

From now on, we will exclusively use the term ``local'' to refer to geometrically local Hamiltonians.
    
    \para{Prior best algorithms}
    To describe the known algorithms for this problem, we need to formally specify the problem. Although our results apply to very general time-dependent Hamiltonians, while comparing to previous work we assume the simpler case where the Hamiltonian is time independent.
    
    We assume our $n$ qubits are laid out in a $D$-dimensional lattice $\Lambda$ in $\mathbb{R}^D$, where $D=\calO(1)$, and every unit ball contains $\calO(1)$ qubits. We assume our Hamiltonian $H$ is given as a sum of terms $H = \sum_{X\subseteq\Lambda} h_X$, where each $h_X$ only acts nontrivally on qubits in $X$ (and acts as identity on the qubits in $\Lambda \setminus X$), such that $h_X=0$ if $\diam(X)>1$,
    which enforces geometric locality. 
    (More formally, we rescale the metric
    in such a way that $\diam(X) > 1$ implies $h_X = 0$.)
    We normalize the Hamiltonian by requiring $\norm{h_X}\leq 1$.
    
    We consider a quantum circuit simulating the time evolution due to such a Hamiltonian efficient if it uses $\poly(n,T,1/\epsilon)$ gates. To get some intuition for what we should hope for, notice that in the real world, time evolution takes time $T$ and uses $n$ qubits. Regarding ``Nature'' as a quantum simulator, we might expect that there is a quantum circuit that uses $\calO(n)$ qubits, $\calO(T)$ circuit depth, and $\calO(nT)$ total gates to solve the problem. It is also reasonable to allow logarithmic overhead in the simulation since such overheads are common even when one classical system simulates the time evolution of another (e.g., when one kind of Turing machine simulates another).
    
    However, previous algorithms for this problem fall short of this expectation. The best Hamiltonian simulation algorithms for sparse Hamiltonians~\cite{TS,QSP,LC16} have query complexity $\calO(nT \polylog(nT/\epsilon))$, but the assumed oracle for the entries requires $\calO(n)$ gates to implement, yielding an algorithm that uses $\calO(n^2T \polylog(nT/\epsilon))$ gates. This was also observed in a recent paper of Childs, Maslov, Nam, Ross, and Su~\cite{Childs2017}, who noted that for $T=n$, all the sparse Hamiltonian simulation algorithms had gate cost proportional to $n^3$ (or worse). A standard application of high-order Lie-Trotter-Suzuki expansions~\cite{Trotter1959,Suzuki1991,Lloyd1996,BAC+07} yields gate complexity $\calO(n^2T (nT/\epsilon)^{\delta})$ for any fixed $\delta > 0$. It has been argued~\cite[Sec.~4.3]{JordanLeePreskill2014} that 
    this in fact yields an algorithm with gate complexity $\calO(nT (nT/\epsilon)^{\delta})$ for any fixed $\delta > 0$. We believe this analysis is correct, but perhaps some details need to be filled in to make the analysis rigorous. In any case, this algorithm still performs  worse than desired, and in particular does not have polylogarithmic dependence on $1/\epsilon$.
    
    \subsection{Results} We exhibit a quantum algorithm that simulates the time evolution due to a time-dependent lattice Hamiltonian with a circuit that uses $\calO(nT \polylog(nT/\epsilon))$ geometrically local $2$-qubit gates (i.e., the gates in our circuit also respect the geometry of the qubits), with depth  $\calO(T \polylog(nT/\epsilon))$ using only $\polylog(nT/\epsilon)$ ancilla qubits. We then also prove a matching lower bound, showing that no quantum algorithm can do better (up to logarithmic factors), even if we relax the output requirement significantly. We now describe our results more formally.
    
    \para{Algorithmic results}
    We consider a more general version of the problem with time-dependent Hamiltonians. In this case we will have $H(t) = \sum_{X\subseteq \Lambda} h_X(t)$ with the locality and norm conditions as before. However, now the operators $h_X(t)$ are functions of time and we need to impose some reasonable constraints on the entries to obtain polynomial-time algorithms.
    
    First we need to be able to compute the entries of our Hamiltonian efficiently at a given time $t$.
    We say that a function $\alpha: [0,T] \ni t \mapsto \alpha(t) \in  \RR$ 
    is \emph{efficiently computable}
    if there is an algorithm that outputs $\alpha(t)$ to precision $\epsilon$ 
    for any given input $t$ specified to precision $\epsilon$
    in running time $\polylog(T/\epsilon)$.
    Note that any complex-valued analytic function 
    on a nonzero neighborhood of a closed real interval in the complex plane
    is efficiently computable (see \cref{app:chebyshev}).
    We will assume that each entry in a local term $h_X(t)$ is efficiently computable.
    
    In addition to being able to compute the entries of the Hamiltonian, we require that the entries do not change wildly with time; otherwise, a sample of entries at discrete times may not predict the behavior of the entries at other times.
    We say that a function $\alpha$ on the interval $[0,T]$ ($T \ge 1$) is 
    \emph{piecewise slowly varying}
    if there are $M = \calO(T)$ intervals $[t_{j-1},t_j]$ 
    with $0 = t_0 < t_1 < \cdots < t_M = T$
    such that $\frac{\mathrm{d}}{\mathrm{d}t}\alpha(t)$ exists and is bounded 
    by $1/(t_j - t_{j-1})$ for $t \in (t_{j-1},t_j)$.
    In particular, a function is piecewise slowly varying if it is a sum of $\calO(T)$ pieces, each of which has derivative at most $\calO(1)$. 
    We will assume that each entry in a term $h_X(t)$ is piecewise slowly varying. 
    
    We are now ready to state our main result, which is proved in \Cref{sec:algorithm}
    
    \begin{restatable}{theorem}{main}
        Let $H(t) = \sum_{X \subseteq \Lambda} h_X(t)$ be a time-dependent Hamiltonian
        on a lattice $\Lambda$ of $n$ qubits, embedded in the Euclidean metric space $\mathbb R^D$.
        Assume that every unit ball contains $\calO(1)$ qubits and	$h_X = 0$ if $\diam(X) > 1$.
        Also assume that every local term $h_X(t)$ is efficiently computable (e.g., analytic), piecewise slowly varying on time domain $[0,T]$, and has  $\norm{h_X(t)} \le 1$ for any $X$ and $t$.
        Then, there exists a quantum algorithm that can approximate the time evolution of $H$
        for time $T$ to accuracy $\epsilon$ using $\calO( Tn \polylog(Tn/\epsilon))$ $2$-qubit local gates,
        and has depth $\calO(T \polylog(Tn/\epsilon))$.
        \label{thm:main}
    \end{restatable}
    
Our algorithm uses $\calO(1)$ ancillas per system qubit on which $H$ is defined.
The ancillas are interspersed with the system qubits,
and all the gates respect the locality of the lattice.
    
\para{Lower bounds}

    We also prove a lower bound on the gate complexity of simulating the time evolution of a time-dependent lattice Hamiltonian. This lower bound matches, up to logarithmic factors, the gate complexity of the algorithm presented in \Cref{thm:main}. Note that unlike previous lower bounds on Hamiltonian simulation~\cite{BAC+07,BerryEtAl2014,BCK15}, which prove lower bounds on query complexity, this is a lower bound on the number of gates required to approximately implement the time-evolution unitary. To our best knowledge, this is the first nontrivial lower bound on the gate complexity of the simulation problem. For concreteness, we focus on a 1-dimensional time-dependent local Hamiltonian in this section, although the lower bound extends to other constant dimensions with minor modifications. The lower bounds are proved in \Cref{sec:lowerbound}.
    
Before stating the result formally, let us precisely define the class of Hamiltonians for which we prove the lower bound.
    We say a Hamiltonian $H(t)$ acting on $n$ qubits is a 
    ``piecewise constant 1D Hamiltonian''
    if $H(t)=\sum_{j=1}^{n-1}H_j(t)$, 
    where $H_j(t)$ is only supported on qubits $j$ and $j+1$ with $\max_t \norm{H_j(t)} = \calO(1)$,
    and there is a time slicing $0 = t_0 < t_1 < \cdots < t_M = T$
    where $t_m - t_{m-1} \le 1$ and $M = \calO(T)$ such that 
    $H(t)$ is time-independent within each time slice.
    
    For such Hamiltonians, the time evolution operator for time $T$ can be simulated with error at most $\epsilon$ 
    using \Cref{thm:main} with $\calO(Tn \polylog(Tn/\epsilon))$ $2$-qubit local gates 
    (i.e., the 2-qubit gates only act on adjacent qubits). 
    In particular, for any constant error, the simulation only requires $\tO(Tn)$ 2-qubit local gates. 
    We prove a matching lower bound, where the lower bound even holds against circuits that may use non-geometrically local (i.e., acting on non-adjacent qubits) $2$-qubit gates from a possibly infinite gate set and unlimited ancilla qubits. 
    
    \begin{restatable}{theorem}{lowerboundgeneral}
        \label{thm:lowerboundgeneral}
        For any integers $n$ and $T \leq 4^n$, 
        there exists a piecewise constant bounded 1D Hamil{-}tonian $H(t)$ on $n$ qubits,
        such that any quantum circuit that approximates the time evolution due to $H(t)$ for time $T$ to constant error must use $\tOmega(Tn)$ $2$-qubit gates. 
        The quantum circuit may use unlimited ancilla qubits and the gates may be non-local and come from a possibly infinite gate set.
    \end{restatable}

    Note that this lower bound only holds for $T\leq 4^n$, because any unitary on $n$ qubits can be implemented with $\tO(4^n)$ 2-qubit local gates~\cite{BBC+95,Kni95}. 
    
    We can also strengthen our lower bound to work in the situation where we are only interested in measuring a local observable at the end of the simulation. The simulation algorithm presented in \Cref{thm:main} provides a strong guarantee: the output state is $\epsilon$-close to the ideal output state in trace distance. Trace distance captures distinguishability with respect to arbitrary measurements, but for some applications it might be sufficient for the output state to be close to the ideal state with respect to local measurements only. We show that even in this limited measurement setting, it is not possible to speed up our algorithm in general. In fact, our lower bound works even if the only local measurement performed is a computational basis measurement on the first output qubit.
    
    \begin{restatable}{theorem}{lowerboundlocal}
        \label{thm:lowerboundlocal}
        For any integers $n$ and $T$ such that $1 \le n\leq T\leq 2^n$, 
        there exists a piecewise constant bounded 1D Hamiltonian $H(t)$ on $n$ qubits, 
        such that any quantum circuit that approximates the time evolution due to $H(t)$ for time $T$ to constant error on any local observable must use $\tOmega(Tn)$ 2-qubit gates. 
        If $T\leq n$, we have a lower bound of $\tOmega(T^2)$ gates. 
        (The quantum circuit may use unlimited ancilla qubits and the gates may be non-local and come from a possibly infinite gate set.)
    \end{restatable}
    
    Note that the fact that we get a weaker lower bound of $\tOmega(T^2)$ when $T\leq n$ is not a limitation, but reflects the fact that small time evolutions are actually easier to simulate when the measurement is local. To see this, consider first simulating the time evolution using the algorithm in \Cref{thm:main}. This yields a circuit with $\tO(Tn)$ 2-qubit local gates. But if we only want the output of a local measurement after time $T$, qubits that are far away from the measured qubits cannot affect the output, since the circuit only consists of 2-qubit local gates. Hence we can simply remove all gates that are more than distance equal to the depth of the circuit, $\tO(T)$, away from the measured qubits. We are then left with a circuit that uses $\tO(T^2)$ gates, matching the lower bound in \Cref{thm:lowerboundlocal}.
    
\subsection{Techniques}
    
    \para{Algorithm}
    Our algorithm is based on a decomposition of the time evolution unitary using Lieb-Robinson bounds~\cite{LiebRobinson1972,Hastings2004LSM,NachtergaeleSims2006,HastingsKoma2006,Hastings2010},
    that was made explicit by Osborne~\cite{Osborne2006} (see also Michalakis~\cite[Sec.~III]{Michalakis2012}),
    which when combined with recent advances in Hamiltonian simulation~\cite{TS,QSP,LC16},
    yields \Cref{thm:main}. 
    
    Lieb-Robinson bounds are theorems that informally state that information travels at a constant speed in geometrically local Hamiltonians. For intuition, consider a 1-dimensional lattice of qubits and a geometrically local Hamiltonian that is evolved for a short amount of time. If the time is too short, no information about the first qubit can be transmitted to the last qubit. Lieb-Robinson bounds make this intuition precise, and show that the qubit at position $n$ is only affected by the qubits and operators at position 1 after time $\Omega(n)$. Note that if this were a small-depth unitary circuit of geometrically local $2$-qubit gates such a statement would follow using a ``lightcone'' argument. In other words, after one layer of geometrically local $2$-qubit gates, the influence of qubit 1 can only have spread to qubit 2. Similarly, after $k$ layers of $2$-qubit gates, the influence of qubit 1 can only have spread up to qubit $k$. The fact that this extends to geometrically local Hamiltonians is nontrivial, and is only approximately true. See \Cref{lem:LRB} for a formal statement of a Lieb-Robinson bound.
    
    We use these ideas to chop up the large unitary that performs time evolution for the full Hamiltonian $H$ into many smaller unitaries that perform time evolution for a small portion of the Hamiltonian. Quantitatively, we break Hamiltonian simulation for $H$ for time $\calO(1)$ into $\calO(n/\log(nT/\epsilon))$ pieces, each of which is a Hamiltonian simulation problem for a Hamiltonian on an instance of size $\calO(\log(nT/\epsilon))$ to exponentially small error. At this point we can use any Hamiltonian simulation algorithm for the smaller piece as long as it has polynomial gate cost and has exponentially good dependence on $\epsilon$. While Hamiltonian simulation algorithms based on product formulas do not have error dependence that is $\polylog(1/\epsilon)$, recent Hamiltonian simulation algorithms, such as \cite{BerryEtAl2014,TS,BCK15,QSP,LC16} have $\polylog(1/\epsilon)$ scaling. 
Thus our result importantly uses the recent advances in Hamiltonian simulation with improved error scaling.
    
    \para{Lower bound} As noted before, we lower bound the gate complexity (or total number of gates) required for Hamiltonian simulation, which is different from prior work which proved lower bounds on the query complexity of Hamiltonian simulation. As such, our techniques are completely different from those used in prior work. Informally, our lower bounds are based on a refined circuit-size hierarchy theorem for quantum circuits, although we are technically comparing two different resources in two different models, which are simulation time for Hamiltonians versus gate cost for circuits. 
    
    As a simple motivating example, consider circuit-size hierarchy theorems for classical or quantum circuits more generally. Abstractly, a hierarchy theorem generally states that a computational model with $X$ amount of a resource (e.g., time, space, gates) can do more if given more of the same resource. For example, it can be shown that for every $G \ll 2^n/n$, there exists a Boolean function on $n$ bits that cannot be computed by a circuit of size $G$, but can be computed by a circuit of size $G+\calO(n)$. We show similar hierarchy theorems for quantum circuit size, except that we show that the circuit of larger size that computes the function actually comes from a weaker family of circuits. Informally, we are able to show that there are functions that can be computed by a larger circuit that uses only geometrically local $2$-qubit gates from a fixed universal gate set that cannot be a computed by a smaller circuit, even if we allow the smaller circuit access to unlimited ancilla bits and non-geometrically local $2$-qubit from an infinite gate set. We then leverage this asymmetric circuit size hierarchy theorem to show that there is a Hamiltonian whose evolution for time $T$ cannot be simulated by a circuit of size $\ll nT$, by embedding the result of any quantum circuit with geometrically local $2$-qubit gates into a piecewise constant Hamiltonian with time proportional to the depth of the circuit.

\subsection{Organization of the paper}

In \cref{sec:algorithm} we analyze our algorithm and prove \cref{thm:main}.
In \cref{sec:lowerbound} we prove the lower bounds of \cref{thm:lowerboundgeneral,thm:lowerboundlocal}.
We conclude in \cref{sec:discussion} with an extended discussion on fermionic systems;
we remark that there exists an embedding of systems with fermions into systems of qubits
and hence our algorithm is applicable.
This paper has several appendices.
In \cref{app:benchmark} 
we report numerical results for system sizes at which our algorithm becomes competitive in terms of actual quantum gate count.
In \cref{app:extensions} we remark how to tailor our algorithm if there is spatial modulation of interaction strength.
We also remark how to reduce the complexity of our algorithm on spatial dimension $D$.
\cref{app:clr} contains a new Lieb-Robinson bound for Hamiltonians that are close to commuting.
This appendix can be read independently and provides a self-contained proof of a Lieb-Robinson bound 
that we use in the main theorem.
\cref{app:chebyshev} explains why analytic functions are efficiently computable;
they often arise when time-dependent Hamiltonians are considered.
\cref{app:qspalg} summarizes what quantum signal processing is 
and contains certain optimization techniques that are used in our numerical benchmark of \cref{app:benchmark}.

\section{Algorithm and analysis}   \label{sec:algorithm}
    
    In this section we establish our main algorithmic result, restated below for convenience:
    
    \begin{figure*}[t]
        \centering
        \includegraphics[width=0.85\textwidth]{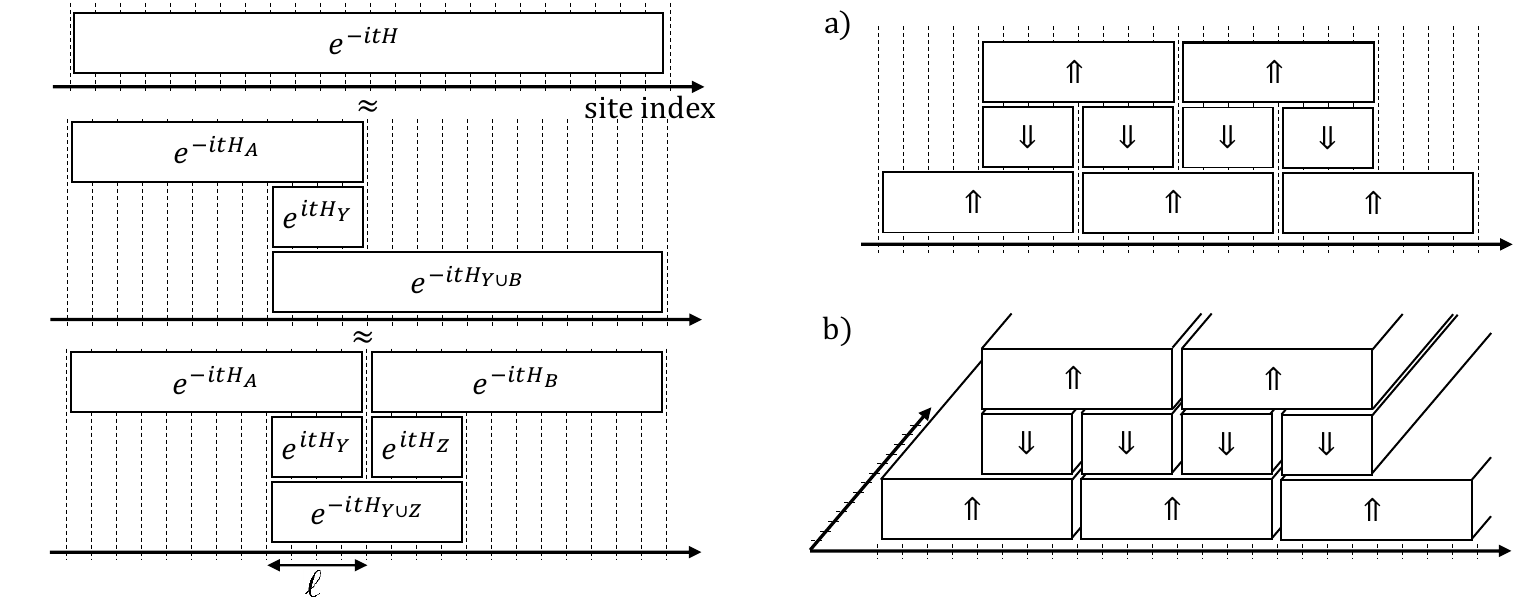}
        \caption{Decomposition of time evolution operator for time $t = \calO(1)$.
            The time is going upwards.
            Each block $\square$ represents the forward time evolution, $e^{-it H_\square}$,
            if the arrow is upward,
            and backward time evolution, $e^{+it H_\square}$, if the arrow is downward.
            Here, $H_\square$ is the sum of local terms in the Hamiltonian 
            supported completely within the block.
            The overlap has size $\ell$.
            (a) shows a one-dimensional setting, 
            but a generalization to higher $D$ dimensions is readily achieved
            by regarding each block as a $(D-1)$-dimensional hyperplane
            so that the problem reduces to lower dimensions.
            (b) shows a two-dimensional setting.
            The approximation error from the depicted decomposition is $\epsilon = \calO(e^{-\mu \ell} L^D / \ell)$
            where $L$ is the linear system size, $\ell$ is the width of the overlap between blocks,
            and $\mu > 0$ is a constant that depends only on the locality of the Hamiltonian.
            One can use any algorithm to further 
            decompose the resulting ``small'' unitaries on $\calO(\log(L/\epsilon))$ qubits
            into elementary gates.
            To achieve gate count that is linear (up to logarithmic factors) in spacetime volume,
            the algorithm for simulating the blocks needs to be polynomial in the block size and polylogarithmic in accuracy.
        }
        \label{fig:algorithm}
    \end{figure*}

    \main*
    
    The algorithm is depicted in \Cref{fig:algorithm}.
    Before showing why this algorithm works, we provide a high-level overview of the algorithm
    and the structure of the proof.
    Since a time evolution unitary $U(T;0)$ is equal to 
    $U(T=t_M; t_{M-1})U(t_{M-1},t_{M-2}) \cdots U(t_2;t_1)U(t_1; t_0 = 0)$,
    we will focus on a time evolution operator $U(t;0)$ where $t = \calO(1)$,
    generated by a slowly varying bounded Hamiltonian.
    The number $M$ of time slices is chosen to be $\calO(T)$.
    The key idea, as shown in \Cref{fig:algorithm}, is that the time-evolution operator, 
    $e^{-itH}$ due to the full Hamiltonian $\sum_{X\subseteq \Lambda} h_X$ 
    can be approximately written as a product 
    \begin{align}
    e^{-itH} \approx 
    \left( e^{-itH_A} \right) \left(   e^{+itH_Y}  \right) \left(  e^{-itH_{Y\cup B}} \right).
    \end{align}
    Here $A\cup B = \Lambda$ and 
    we think of $A$ and $B$ as large regions, but $Y$ as a small subset of $A$.
    The error in the approximation is exponentially small in the diameter of $Y$.
    This is formally proved in \Cref{lem:decomposition}, which is supported by 
    \Cref{lem:ode} and \Cref{lem:LRB}. Applying this twice, using 
    \begin{align}
    e^{-itH_{Y\cup B}} 
    \approx 
    \left( e^{-itH_B} \right)
    \left( e^{+itH_Z}  \right)
    \left( e^{-itH_{Y \cup Z}} \right)
    \end{align}
    leads to a symmetric approximation
    as depicted at the bottom left of \Cref{fig:algorithm}.
    This procedure can then be repeated
    for the large operators supported on $A$ and $B$ to reduce the size of all the operators
    involved, leading to the pattern in \Cref{fig:algorithm} (a).
    This reduces the problem of
    implementing the time-evolution operator for $H$ into the problem of implementing smaller
    time-evolution operators, which can be implemented using known quantum algorithms.%
	\footnote{
	The condition on the Hamiltonian that each term be slowly varying with respect to time,
	is only needed to use Hamiltonian simulation algorithms of \cite{BerryEtAl2014,TS,Low2018,Kieferova2018}.
	If there is a Hamiltonian simulation algorithm whose gate count is polylogarithmic in accuracy
	and polynomial in system size and evolution time regardless of the time derivative,
	then our results will give an algorithm that is independent of the time derivative as well.
	}

    We now establish the lemmas needed to prove the result.
    
    \begin{lemma}
    Let $A_t$ and $B_t$ be continuous time-dependent Hermitian operators,
    and let $U^A_t$ and $U^B_t$ with $U^A_0 = U^B_0 = \id$ 
    be the corresponding time evolution unitaries.
    Then the following hold:
    \begin{enumerate}
    \item[(i)] $W_t=(U^B_t)^\dagger U^A_t$ is the unique solution of 
    $\i\partial_t W_t = \left((U^B_t)^\dagger (A_t-B_t) U^B_t \right) W_t$ 
    and $W_0 = \id$.
    \item[(ii)] If $\norm{A_s-B_s} \le \delta$ for all $s \in [0,t]$,
     then $\norm{U^A_t - U^B_t} \le t \delta$.
    \end{enumerate}
    \label{lem:ode}
    \end{lemma}
    \begin{proof} 
    (i) Differentiate.
    The solution to the ordinary differential equation is unique.
    (ii) Apply Jensen's inequality for $\norm{\cdot}$ 
    (implied by the triangle inequality for $\norm{\cdot}$) 
    to the equation $W_t - W_0 = \int_0^t \rd s \partial_s W_s$.
    Then, invoke (i) and the unitary invariance of $\norm{\cdot}$.
    \end{proof}

For any two sites $x,y$ we denote by $\dist(x,y)$ the distance between two sites $x,y$,
and for any two sets $X,Y$ of sites we write $\dist(X,Y) = \min_{x \in X, y \in Y} \dist(x,y)$.
The diameter of a set $X$ is $\diam(X) = \max_{x, x' \in X} \dist(x,x')$.
\begin{lemma}[Lieb-Robinson 
    bound~\cite{LiebRobinson1972,Hastings2004LSM,NachtergaeleSims2006,HastingsKoma2006}]
\label{lem:LRB}
    Let $H=\sum_X h_X$ be a local Hamiltonian and 
    $O_X$ be any operator supported on $X$,
    and put $\ell = \lfloor \dist(X,\Lambda\setminus \Omega) \rfloor$.
    Then
\begin{align}
    \norm{(U^{H}_t)^\dagger O_X U^{H}_t
        - (U^{H_{ \Omega}}_t)^\dagger O_X U^{H_{ \Omega}}_t}
    \le 
    |X|\norm{O_X}
    \frac{(2\zeta_0 |t|)^\ell}{\ell!},
\label{eq:LRB}
\end{align}
    where $\zeta_0 = \max_{p \in \Lambda} \sum_{Z \ni p} |Z|\norm{h_Z} = \calO(1)$.
    In particular, there are constants $v_{LR} > 0$, called the Lieb-Robinson velocity,%
    \footnote{
    Strictly speaking, the Lieb-Robinson velocity 
    is defined to be the infimum of any $v_{LR}$ such that \cref{eq:exteriorLightcone} holds.
    }
    and $\mu>0$, 
    such that for $\ell \geq v_{LR} |t|$, we have
\begin{align}
\norm{(U^{H}_t)^\dagger O_X U^{H}_t
      - (U^{H_{ \Omega}}_t)^\dagger O_X U^{H_{ \Omega}}_t}
\le
\calO(|X| \norm{O_X} \exp(- \mu \ell)).
\label{eq:exteriorLightcone}
\end{align}
\end{lemma}
\begin{proof}
        See \ref{app:pfLRB}.
\end{proof}

    We are considering strictly local interactions (as in \Cref{thm:main}),
    where $h_X = 0$ if $\diam(X)>1$,
    but similar results hold with milder locality conditions such as $\norm{h_X} \le e^{-\diam(X)}$~\cite{LiebRobinson1972,Hastings2004LSM,NachtergaeleSims2006,HastingsKoma2006,Hastings2010};
    see \cref{app:clr} for a detailed proof.
    Below we will only use the result that 
    the error is at most $\calO(e^{ - \mu \ell})$ for some $\mu > 0$ and fixed $t$.
    For slower decaying interactions, the bound is weaker and
    the overlap size $\ell$ in \Cref{fig:algorithm} will have to be larger.
    
    The Lieb-Robinson bound implies the following decomposition.
    \begin{lemma}
    Let $H = \sum_X h_X$ be a local Hamiltonian (as in \Cref{thm:main}, or a more general definition
    for which \Cref{lem:LRB} still holds).
    Then there is a constant $\mu >0$ such that
    for any disjoint regions $A,B,C$, and for constant $t$, we have
\begin{align}
      \norm{U_t^{H_{A\cup B}} (U_{t}^{H_{B}})^\dagger U_t^{H_{B \cup C}} 
        - U_t^{H_{A\cup B \cup C}}} \le 
      \calO( e^{- \mu \dist(A,C)} ) \sum_{X: \text{bd}(AB,C)}\norm{h_X},
\end{align}
    where $X:\text{bd}(AB,C)$ means that $X \subseteq A \cup B \cup C$ 
    and $X \not\subseteq A \cup B$  and $X \not\subseteq C$.
    \label{lem:decomposition}
    \end{lemma}
A similar technique of the proof below has appeared in \cite{Osborne2006,Michalakis2012}.
In these references one approximates $U_t^{H_{A\cup B}}$ by $U_t^{H_A}U_t^{H_B} V_t$
where $V_t$ is generated%
\footnote{
    We say a unitary $U_t$ is generated by $H_t$ 
    if $U_t$ is the solution to $\i\partial_t U_t = H_t U_t$ with $U_0 = \id$.
}
by some time dependent Hermitian operator of small support.
    \begin{proof}
    We omit ``$\cup$'' for the union of disjoint sets. 
    The following identity is trivial but important:
    \begin{align}
        U^{H_{ABC}}_t =  
        U^{H_{AB} + H_{C}}_t \underbrace{(U^{H_{AB} + H_{C}}_t )^\dagger  U^{H_{ABC}}_t}_{=W_t}.
    \end{align}
    By \Cref{lem:ode} (i), $W_t$ is generated by the Hamiltonian
    \begin{align}
    &(U^{H_{AB} + H_{C}}_t)^\dagger 
    (\underbrace{H_{ABC}-H_{AB}-H_C}_{H_\text{bd}}) 
    U^{H_{AB} + H_{C}}_t  = 
    (U^{H_{AB} + H_{C}}_t)^\dagger 
    H_\text{bd} 
    U^{H_{AB} + H_{C}}_t. \nonumber 
    \end{align}
    Applying \Cref{lem:LRB} (and \cref{eq:exteriorLightcone} in particular) with $O_X = H_\text{bd}$, we have
    \begin{align}
    \norm{(U^{H_{AB} + H_{C}}_t)^\dagger H_\text{bd} U^{H_{AB} + H_{C}}_t -
    (U^{H_{B} + H_{C}}_t)^\dagger H_\text{bd} U^{H_{B} + H_{C}}_t} \leq 
    \underbrace{\calO(\norm{H_\text{bd}} \exp(-\mu \ell))}_{=\delta}, \label{eq:truncation}
    \end{align}
    for some $\mu>0$ where $\ell$ is the distance between the support of the boundary terms 
    $H_\text{bd}$ and $A$.
    Since $H_\text{bd}$ contains terms that cross between $AB$ and $C$,
    the distance $\ell$ is at least $\dist(A,C)$ minus~2. 
    Using this and $\norm{H_\text{bd}}\leq \sum_{X: \text{bd}(AB,C)}\norm{h_X}$, we get $\delta = \calO(e^{-\mu \dist(A,C)} \sum_{X: \text{bd}(AB,C)}\norm{h_X})$.

    We can now use \Cref{lem:ode} (i) again to compute the unitary generated by the second term in \eqref{eq:truncation}, $(U^{H_{B} + H_{C}}_t)^\dagger H_\text{bd} U^{H_{B} + H_{C}}_t$.
    The unitary generated is $(U^{H_{B} + H_{C}}_t)^\dagger U^{H_\text{bd}+H_B+H_C} = (U^{H_{B} + H_{C}}_t)^\dagger U^{H_{BC}}_t$,
    where we used the fact that $H_{ABC}-H_{AB}+H_B = H_{BC}$. 
    This operator can be thought of as the ``interaction picture'' time-evolution operator
    of the second term in \eqref{eq:truncation}.
    This is our simplification of the ``patching'' unitary.
    Applying \Cref{lem:ode} (ii), we get 
    \begin{align*}
    \norm{W_t - (U^{H_{B} + H_{C}}_t)^\dagger U^{H_{BC}}_t} \leq t\delta = O(\delta)
    \end{align*}
    The left-hand side can be further simplified as
    \begin{align*}
        \norm{W_t - (U^{H_{B} + H_{C}}_t)^\dagger U^{H_{BC}}_t} 
        &= \norm{(U^{H_{AB} + H_{C}}_t )^\dagger  U^{H_{ABC}}_t - (U^{H_{B} + H_{C}}_t)^\dagger U^{H_{BC}}_t} \\
        &= \norm{ U^{H_{ABC}}_t - (U^{H_{AB} + H_{C}}_t )(U^{H_{B} + H_{C}}_t)^\dagger U^{H_{BC}}_t} \\
        &= \norm{ U^{H_{ABC}}_t - (U^{H_{AB}}_t U^{H_{C}}_t )(U^{H_{B}}_t U^{H_{C}}_t)^\dagger U^{H_{BC}}_t}\\
        &= \norm{ U^{H_{ABC}}_t - (U^{H_{AB}}_t)(U^{H_{B}}_t)^\dagger U^{H_{BC}}_t}.        
        \end{align*}
    This yields the desired inequality in the statement of the lemma.
    \end{proof}
    
\begin{proof}[Proof of \protect{\Cref{thm:main}}]
    The circuit for simulating the Hamiltonian is described in \Cref{fig:algorithm}.
    The decomposition of the time evolution unitary in \Cref{fig:algorithm}
    is the result of iterated applications of \Cref{lem:decomposition}.
    For a one-dimensional chain with open boundary condition,
    let $L$ be the length of the chain so there are $\calO(L)$ qubits.
    Take two blocks $A$ and $Y\cup B$ of the chain 
    such that their intersection $Y$ has length $\ell \ll L$
    and their union is the whole chain.
(See \cref{fig:algorithm}.)
    Applying \Cref{lem:decomposition},
    we approximate the full unitary by the composition of forward time evolution on $Y \cup B$,
    backward time evolution on $Y$,
    and forward time evolution on $A$.
    This incurs approximation error $\delta$ in spectral norm. 
    Every block unitary in this decomposition is a time evolution operator
    with respect to the sum of Hamitonian terms within the block,
    and we can recursively apply the decomposition for large blocks (size $\gg \ell$).
    We end up with a layout of small unitaries as shown in \Cref{fig:algorithm} (a).
    The error from this decomposition is $\calO(\delta L/ \ell)$,
    which is exponentially small in $\ell$ for $t = \calO(1)$.

    Going to higher dimensions $D > 1$,
    we first decompose the full time evolution 
    into unitaries on $\calO(L/\ell)$ slabs of size $L^{D-1} \times 2\ell$.
    (Since $L \gg \ell$, each slab looks like a hyperplane of codimension~$1$.)
    This entails error $\calO(e^{-\mu \ell} L^D / \ell)$ 
    since the boundary term has norm at most $\calO(L^{D-1})$.
    For each slab the decomposition into $\calO(L/\ell)$ blocks of size $L^{D-2} \times 2\ell \times 2\ell$,
    which look like hyperplanes of codimension~$2$,
    gives error $\calO(e^{-\mu \ell} (\ell L^{D-2}) (L/\ell))$.
    Summing up all the (thickened) hyperplanes,
    we get $\calO(e^{-\mu \ell} L^{D}/\ell)$ for the second round of decomposition.
    After $D$ rounds of the decomposition the total error is $\calO(e^{-\mu \ell} D L^{D}/\ell)$,
    and we are left with $\calO((L/\ell)^D)$ blocks of unitaries for $t = \calO(1)$.

    For longer times, apply the decomposition to each factor of 
\begin{align*}    
    U(T=t_M; t_{M-1})U(t_{M-1},t_{M-2}) \cdots U(t_2;t_1)U(t_1; t_0 = 0).
\end{align*}
    
    It remains to implement the unitaries on $m = \calO(T L^D/\ell^D)$ blocks $\square$ of $\calO(\ell^D)$ qubits
    where $\ell = \calO(\log (T L/\epsilon))$ to accuracy $\epsilon / m$.
    All blocks have the form $U^{H_{ \square}}_t$, and can be implemented using any known
    Hamiltonian simulation algorithm.
    For a time-independent Hamiltonian, 
    if we use an algorithm that is polynomial in the spacetime volume
    and polylogarithmic in the accuracy 
    such as those based on signal processing~\cite{QSP,LC16} or 
    linear combination of unitaries~\cite{BerryEtAl2014,TS,BCK15},
    then the overall gate complexity is 
    $\calO(T L^D \polylog(T L / \epsilon)) = \calO(Tn \polylog(Tn/\epsilon))$,
    where the exponent in the $\polylog$ factor depends on the choice of the algorithm.%
    \footnote{
    If we use the quantum signal processing based algorithms~\cite{QSP,LC16}
    to implement the blocks of size $\calO(\ell^D)$,
    then we need $\calO(\log \ell)$ ancilla qubits for a block.
    Thus, if we do not mind implementing them all in serial,
    then it follows that the number of ancillas needed is 
    $\calO(\log \log (T n / \epsilon))$,
    which is much smaller than what would be needed if
    the quantum signal processing algorithm was 
    directly used to simulate the full system.
    }
    Similarly for a slowly varying time-dependent Hamiltonian, 
    we achieve the same gate complexity by using any time-dependent Hamiltonian simulation algorithm
    that is polynomial in the spacetime volume and polylogarithmic in the accuracy.
    For example, the fractional queries algorithm~\cite{BerryEtAl2014} or the Taylor series approach~\cite{TS,Low2018,Kieferova2018} possess these properties.
    \end{proof}
    For not too large system sizes $L$,
    it may be reasonable to use a bruteforce method to decompose the block unitaries 
    into elementary gates~\cite[Chap. 8]{KitaevBook}.

\section{Optimality}    \label{sec:lowerbound}
    
In this section we prove a lower bound on the gate complexity 
of the problem of simulating the time evolution of a time-dependent local Hamiltonian. 
(Recall that throughout this paper we use \emph{local} to mean geometrically local.) 
    
\subsection{Lower bound proofs}
    
We now prove \Cref{thm:lowerboundgeneral} and \Cref{thm:lowerboundlocal}, 
starting with \Cref{thm:lowerboundlocal}. 
This lower bound follows from the following three steps. 
First, in \Cref{lem:circuittoHamiltonian}, 
we observe that for every depth-$T$ quantum circuit on $n$ qubits that uses  local $2$-qubit gates, 
there exists a piecewise constant bounded Hamiltonian $H(t)$
such that time evolution due to $H(t)$ for time $T$ is equal to applying the quantum circuit. 
Then, in \Cref{lem:countlower} we show that the number of distinct Boolean functions on $n$ bits 
computed by such quantum circuits is at least exponential in $\tOmega(Tn)$, 
where we say a quantum circuit has computed a Boolean function 
if its first output qubit is equal to the value of the Boolean function with high probability. 
Finally, in \Cref{lem:countupper} we observe that the maximum number of Boolean functions 
that can be computed (to constant error) by the class of quantum circuits 
with $G$ arbitrary non-local $2$-qubit gates from any (possibly infinite) gate set 
is exponential in $\tO(G\log n)$. 
Since we want this class of circuits to be able to simulate all piecewise constant bounded 1D Hamiltonians for time $T$, 
we must have $G = \tOmega(Tn)$.
    
\begin{lemma}\label{lem:circuittoHamiltonian}
Let $U$ be a depth-$T$ quantum circuit on $n$ qubits that uses local $2$-qubit gates from any gate set. 
Then there exists a piecewise constant bounded 1D Hamiltonian $H(t)$ 
such that the time evolution due to $H(t)$ for time $T$ exactly equals $U$. 
\end{lemma}
\begin{proof}
We first prove the claim for a depth-$1$ quantum circuit. 
This yields a Hamiltonian $H(t)$ that is defined for $t\in[0,1]$, 
whose time evolution for unit time equals the given depth-$1$ circuit. Then we can apply the same argument to each layer of the circuit, obtaining Hamiltonians valid for times $t\in[1,2]$, and so on, until $t\in[T-1,T]$. This yields a Hamiltonian $H(t)$ defined for all time $t\in[0,T]$ whose time evolution for time $T$ equals the given unitary. If the individual terms in a given time interval have bounded spectral norm, then so will the Hamiltonian defined for the full time duration. 
        
For a depth $1$ circuit with local $2$-qubit gates, since the gates act on disjoint qubits we only need to solve the problem for one $2$-qubit unitary and sum the resulting Hamiltonians. Consider a unitary $U_j$ that acts on qubits $j$ and $j+1$. By choosing $H_j = \i \log U_j$, we can ensure that $e^{-iH_j} = U_j$ and $\norm{H_j} = \calO(1)$. 
        
The overall Hamiltonian is now piecewise constant with $T$ time slices.
\end{proof}
    
Note that it also possible to use a similar construction to obtain a Hamiltonian 
that is continuous (instead of being piecewise constant) 
with a constant upper bound on the norm of the first derivative of the Hamiltonian.
One way to achieve this is 
to make the Hamiltonian to be zero for integer values of time;
that is, $H(t)$ is not identically zero, but is zero for $t \in \ZZ \subset \RR$.
Concretely, let $g(\tau) = 6\tau(1-\tau)$ be a real function
which satisfies $g(0) = g(1) = 0$ and $\int_0^1 g(\tau) \rd \tau = 1$.
We then let $H_j(t) = \i g(t-v+1) \log U_j$ for a single two-qubit unitary $U_j$
that is in the $v$-th layer of the circuit where $v = 1,2,\ldots,T$.

    \begin{lemma}
        \label{lem:countlower}
        For any integers $n$ and $T$ such that $2 \le n \leq T \leq 2^n$, the number of distinct Boolean functions $f:\{0,1\}^n\to \{0,1\}$ that can be computed by depth-$T$ quantum circuits on $n$ qubits that use local $2$-qubit gates from a finite gate set is at least $2^{\tOmega(Tn)}$.
    \end{lemma}
    \begin{proof}
        We first divide the $n$ qubits into groups of $k=\lfloor \log_2 T \rfloor$ qubits,
         which is possible since $T \leq 2^n$. 
        On these $k$ qubits, we will show that it is possible to compute $2^{\tOmega(T)}$ distinct Boolean functions 
        with a depth $T$ circuit that uses local $2$-qubit gates. 
        One way to do this is to consider all Boolean functions on $k'<k$ bits. 
        Any Boolean function $f_x$ that evaluates to $f_x(x)=1$ on exactly one input $x$ of size $k'$ 
        can be computed with a circuit of $\tO(k')$ gates and $\tO(k')$ depth 
        using only $2$-qubit gates and $1$ ancilla qubit in addition to one output qubit~\cite[Corollary 7.4]{BBC+95}.
        An arbitrary Boolean function $f:\{0,1\}^{k'} \to \{ 0,1\}$ is a sum of such functions: 
        $f = \sum_{x \in f^{-1}(1)} f_x = \bigoplus_{x \in f^{-1}(1)} f_x$.
        Implementing all $f_x$ for $x \in f^{-1}(1)$ in serial using a common output qubit,
        we obtain a circuit for the full function $f$.
        Since $f^{-1}(1)$ consists of at most $2^{k'}$ bit strings, 
        this will yield a circuit of size $\tO(2^{k'})$ and depth $\tO(2^{k'})$.
        Note that each of the $2$-qubit gates may be made local without changing these expressions by more than a log factor in the exponent --  $\mathcal{O}(k')$ local pairwise SWAP gates suffice to move any two target qubits next two each other.
        By choosing $k' = k - \Theta(\log k)$,
        we can compute all Boolean functions on $k'$ bits with depth at most $T$. 
        Since there are $2^{2^{k'}} = 2^{\tOmega(T)}$ distinct Boolean functions on $k'$ bits, 
        we have shown that circuits with depth $T$ using $k=\lfloor \log_2 T \rfloor$ qubits can compute at least $2^{\tOmega(T)}$ distinct Boolean functions.
        
        We can compute $2^{\tOmega(T)}$ distinct Boolean functions on each of the $n/k$ blocks of $k$ qubits to obtain $(2^{\tOmega(T)})^{n/k} = 2^{\tOmega(Tn)}$ distinct Boolean functions with $n/k$ outputs. I.e., we have computed a function $\{0,1\}^n \to \{0,1\}^{n/k}$. Since we want to obtain a single-output Boolean function, as the overall goal is to prove lower bounds against simulation algorithms correct on local measurements, we combine these Boolean functions into one. We do this by computing the parity of all the outputs of these $n/k$ Boolean functions using CNOT gates. Computing the parity uses at most $n$ 2-qubit local gates and has depth $n$. The circuit now has depth $T + n \leq 2T$ and by rescaling $T$ we can make this circuit have depth $T$, while retaining the lower bound of $2^{\tOmega(Tn)}$ distinct Boolean functions.
        
        Unfortunately, after taking the parity of these $n/k$ functions, it is not true that the resulting functions are all distinct. For example, the parity of functions $f(x)$ and $g(y)$ is a new function $f(x) \oplus g(y)$, which also happens to be the parity of the functions $\neg f(x)$ and $\neg g(y)$. To avoid this overcounting of functions, we do not use all possible functions on $k'$ bits in the argument above, but only all those functions that map the all-zeros input to 0. This only halves the total number of functions, which does not change the asymptotic expressions above. With this additional constraint, it is easy to see that if $f(x) \oplus g(y) = f'(x) \oplus g'(y)$, this implies that $f$ and $f'$ are the same, by fixing $y$ to be the all-zeros input, and similarly that $g$ and $g'$ are the same.
    \end{proof}

    We say a quantum circuit $U$ computes a Boolean function $f:\{0,1\}^n\to\{0,1\}$ with high probability if measuring the first output qubit of $U|x_1x_2\cdots x_n0\cdots 0\rangle$ yields $f(x)$ with probability at least $2/3$.
    \begin{lemma}
        \label{lem:countupper}
        The number of Boolean functions $f:\{0,1\}^n\to \{0,1\}$ that can be computed with high probability by quantum circuits with unlimited ancilla qubits using $G$ non-local $2$-qubit gates from any gate set is at most $2^{\tO(G\log n)}$.
    \end{lemma}
    \begin{proof}
    First we note that if a circuit $U$ with $G$ arbitrary $2$-qubit gates from any gate set computes a Boolean function with probability at least $2/3$, then there is another circuit with $\tO(G)$ gates from a finite 2-qubit non-local gate set that computes the same function with probability at least $2/3$. We do this by first boosting the success probability of the original circuit using an ancilla qubit to a constant larger than $2/3$ and then invoking the Solovay--Kitaev theorem \cite{NC00} to approximate each gate in this circuit to error $\calO(1/G)$ with a circuit from a finite gate set of $2$-qubit gates. This increases the circuit size to $\tO(G)$ gates. Since each gate has error $\calO(1/G)$, the overall error is only a constant, and the new circuit computes the Boolean function $f$ with high probability.
        
        We now have to show that the number of Boolean functions on $n$ bits computed by a circuit with $\tO(G)$ non-local 2-qubit gates from a finite gate set is at most $2^{\tO(G\log n)}$. To do so, we simply show that the total number of distinct circuits with $\tO(G)$ non-local 2-qubit gates from a finite gate set is at most $2^{\tO(G \log n)}$.
        
        First observe that a circuit with $\tO(G)$ gates can only use $\tO(G)$ ancilla qubits, since each $2$-qubit gate can interact with at most 2 new ancilla qubits. Furthermore, the depth of a circuit cannot be larger than the number of gates in the circuit. Let us now upper bound the total number of quantum circuits of this form. Each such circuit can be specified by listing the location of each gate and which gate it is from the finite gate set. Specifiying the latter only needs a constant number of bits since the gate set is finite, and the location can be specified using the gate's depth, and the labels of the two qubits it acts on. The depth only requires $\calO(\log G)$ bits to specify, and since there are at most $n+\tO(G)$ qubits, this only needs $\calO(\log n + \log G)$ bits to specify. In total, since there are $\tO(G)$ gates, the entire circuit can be specified with $\tO(G\log n)$ bits. Finally, since any such circuit can be specified with $\tO(G\log n)$ bits, there can only be $2^{\tO(G \log n)}$ such circuits.
    \end{proof}
    
    \begin{proof}[Proof of \protect{\Cref{thm:lowerboundlocal}}]
        Suppose that any piecewise constant bounded 1D Hamiltonian on $n$ qubits 
        can be simulated for time $T$ using $G$ 2-qubit non-local gates from any (possibly infinite) gate set. 
        Then using \Cref{lem:circuittoHamiltonian} and \Cref{lem:countlower}, 
        we can compute at least $2^{\tOmega(Tn)}$ distinct Boolean functions using such Hamiltonians. 
        By assumption, each of these Boolean functions can be approximately computed by a circuit with $G$ gates. 
        Now invoking \Cref{lem:countupper}, 
        we know that such circuits can compute at most $2^{\tO(G\log n)}$ $n$-bit Boolean functions. 
        Hence we must have $G\log n = \tOmega(Tn)$, which yields $G = \tOmega(Tn)$.
        
        The proof for $T\leq n$ follows in a black-box manner from the first part of the theorem statement by only using $T$ out of the $n$ available qubits. In this case the first part of the theorem applies and yields a lower bound of $\tOmega(T^2)$.
    \end{proof}
    
We now prove \Cref{thm:lowerboundgeneral}, which follows a similar outline.
The first step is identical, and we can reuse \Cref{lem:circuittoHamiltonian}.
For the next step, instead of counting distinct Boolean functions, 
we count the total number of ``distinct'' unitaries.
Unlike Boolean functions on $n$ bits, 
there are infinitely many  unitaries on $n$ qubits. 
Hence we count unitaries that are ``distinguishable.'' 
Formally, we say $U$ and $V$ are distinguishable 
if there is a state $\ket \psi$ 
such that $U\ket \psi$ and $V \ket \psi$ have trace distance, say,~$0.1$.
In \Cref{lem:Ucountlower} we show that the number of distinguishable unitaries 
computed by quantum circuits on $n$ qubits with depth $T$ is exponential in $\tOmega(Tn)$.
As before, we then show that the maximum number of distinguishable unitaries 
that can be computed (to constant error) by the class of quantum circuits 
with $G$ arbitrary non-local $2$-qubit gates from any (possibly infinite) gate set 
is exponential in $\tO(G\log n)$.
    
\begin{lemma} \label{lem:Ucountlower}
For any integers $n,T$ such that $4 \le n \le T \le 4^n$, 
there exists a set of unitaries of cardinality $2^{\tOmega(Tn)}$
such that every unitary in the set can be computed by a depth-$T$ quantum circuit on $n$ qubits 
that uses local $2$-qubit gates from a finite gate set 
and any $U \neq V$ from this set are distinguishable. 
\end{lemma}
\begin{proof}
We divide the $n$ qubits into groups of $k=\lfloor \log_4 T \rfloor$ qubits, 
which is possible since $T \leq 4^n$. On these $k$ qubits, we will 
compute $2^{\tOmega(T)}$ distinguishable unitaries 
with a depth $T$ circuit that uses local $2$-qubit gates. 
We can do this by considering a maximal set of unitaries on $k'$ qubits that is distinguishable. 
More precisely, on $k'$ qubits there exist $2^{\Omega(4^{k'})}$ unitaries 
such that each pair of unitaries is at least distance $0.1$ apart in spectral norm; see e.g.~\cite{Szarek1997}. 
(This follows from the fact that in the group of $d \times d$ unitaries 
with metric induced by operator norm, a ball of radius~$0.1$ has volume exponentially small in $d^2$.)
If $\norm{ U_1 \ket \psi - U_2 \ket \psi }_2 \ge 0.1$,
then the trace distance between the two states 
$\frac{1}{\sqrt{2}}(\ket0 \ket 0 + \ket 1 U_j \ket \psi)$ ($j=1,2$)
is $\ge 0.1$,
and hence the controlled unitaries $\ket 0 \bra 0 \otimes I + \ket 1 \bra 1 \otimes U_j$ ($j=1,2$) 
are distinguishable.
Therefore, on $k'$ qubits there exist $2^{\Omega(4^{k'})}$ unitaries that are pairwise distinguishable.
We know that any unitary on $k'$ qubits can be exactly written as a product of 
$\tO(4^{k'})$ arbitrary 2-qubit gates~\cite{BBC+95,Kni95}.
As described in the proof of~\cref{lem:countlower}, 
making these gates local and from a finite gate set only adds polynomial factors in $k'$.
By choosing $k' = k - \Theta(\log k)$, 
we can compute this set of $2^{\Omega(4^{k'})} = 2^{\tOmega(T)}$ distinguishable unitaries 
with depth at most $T$. 

If $U$ and $V$ are distinguishable, 
then so are $U \otimes X$ and $V \otimes Y$ for any unitary $X,Y$,
since the distinguisher can simply ignore the second register.
Hence, if we have two sets $\{U_i\}_{i=1}^q, \{V_j\}_{j=1}^q$ of distinguishable unitaries,
the set $\{U_i \otimes V_j \}_{i,j=1}^q$ consists of $q^2$ distinguishable unitaries.
Since we can compute $2^{\tOmega(T)}$ distinguishable unitaries on each of the $n/k$ blocks of $k$ qubits,
we can compute $(2^{\tOmega(T)})^{n/k} = 2^{\tOmega(Tn)}$ unitaries on all $n$ qubits.
\end{proof}
    
\begin{lemma}\label{lem:Ucountupper}
Let $S$ be a set of pairwise distinguishable unitaries. 
If any unitary in $S$ can be computed by a quantum circuit 
with $G$ non-local $2$-qubit gates from any gate set, then $|S| = 2^{\tO(G\log n)}$.
\end{lemma}
\begin{proof}
This proof is essentially the same as that of \Cref{lem:countupper}.
We first observe that if a circuit over an arbitrary gate set computes a unitary $U$, 
we can approximate it to error less than $0.04$ using the Solovay--Kitaev theorem 
and increase the circuit size to $\tO(G)$ gates. 
Importantly, since the unitaries are a distance $0.1$ apart (see \cref{eq:cbnorm-opnorm}),
one circuit cannot simultaneously approximate two different unitaries to error $0.04$.
Then exactly the same counting argument as in \Cref{lem:countupper} 
shows there can only be $2^{\tO(G \log n)}$ such circuits.
\end{proof}
    
\begin{proof}[Proof of \protect{\Cref{thm:lowerboundgeneral}}]
        Suppose that any piecewise constant bounded local Hamiltonian on $n$-qubits could be simulated for time $T$ using $G$ 2-qubit non-local gates from any (possibly infinite) gate set. Then using \Cref{lem:circuittoHamiltonian} and \Cref{lem:Ucountlower}, we can produce a set $S$ of distinguishable unitaries of size $2^{\tOmega(Tn)}$. By assumption, each of these unitaries can be approximately computed by a circuit with $G$ non-local 4-qubit gates. Now invoking \Cref{lem:countupper}, we know that such circuits can approximate at most $2^{\tO(G\log n)}$ distinguishable unitaries on $n$ qubits. Hence we must have $G\log n = \tOmega(Tn)$, which yields $G = \tOmega(Tn)$.
\end{proof}
    
\section{Discussion}\label{sec:discussion}
    
    We have only analyzed local Hamiltonians on (hyper)cubic lattices embedded in some Euclidean space,
    but Lieb-Robinson bounds with exponential dependence on the separation distance
    hold more generally.  However, on more general graphs, it may be more difficult to find an appropriate decomposition that gives a small error; this must be analyzed for each graph.
    One advantage of the method here is that the accuracy improves for smaller Lieb-Robinson velocity.
    This can occur if the terms in the Hamiltonian have a small commutator (see \cref{app:clr}).
    
    The decomposition based on Lieb-Robinson bounds looks 
    very similar to higher order Lie-Trotter-Suzuki formulas.
    The difference is in the fact that the overlap 
    is chosen to be larger and larger
    (though very slowly) as the simulated spacetime volume increases.
    If we want an algorithm that does not use any ancilla qubits,
    similar to algorithms based on Lie-Trotter-Suzuki formulas,
    then we can simulate the small blocks from Lieb-Robinson bounds
    by high order Suzuki formulas~\cite{Suzuki1991,BAC+07} 
    where the accuracy dependence is polynomial (power-law) of arbitrarily small exponent $a >0$.
    This combination results in an algorithm of total gate complexity $e^{\calO(1/a)}\cdot \calO(Tn(Tn/\epsilon)^a)$,
    similar to what is claimed to be achievable in Ref.~\cite[Sec.~4.3]{JordanLeePreskill2014}.
    However, a poly-logarithmic dependence on the factor $Tn/\epsilon$ is not possible in this approach
    by any choice of $a$ due to the exponential prefactor.
    
Application to fermions, which represent common physical particles such as electrons, 
is straightforward but worth mentioning.
While we have previously focused on Hamiltonians that act on qubits, 
we now consider Hamiltonians on $n$ sites that are occupied by fermions, 
and describe their well-known reduction to qubit Hamiltonians.
Each term in a fermion Hamiltonian is a product of some number of fermion operators $c_{j}$
and their Hermitian conjugates, indexed by the site $j\in\{0,1,\cdots,n-1\}$, e.g. $c_{0}c^\dagger_{1}$.
These fermion operators are defined by the anti-commutation relations
\begin{align*}
	\{c_j,c_k\}:= c_j c_k + c_k c_j = 0,\quad 	\{c^\dagger_j,c^\dagger_k\}=0, \quad 		\{c_j,c^\dagger_k\}=\delta_{jk}\id.
\end{align*}
Since Hamiltonian terms always have fermion parity even
(this is a basic assumption on any physical Hamiltonian),
Lieb-Robinson bounds hold without any change.
It is often convenient to additionally represent fermion operators 
by (real) Majorana fermion operators $\gamma_{j}$ where $j\in\{0,1,\cdots 2n-1\}$, 
defined by the linear relations
\begin{align*}
    c_{p} = \frac{\gamma_{2p} + \i\gamma_{2p+1}}{2}, \quad 
    c_{p}^\dagger = \frac{\gamma_{2p} - \i\gamma_{2p+1}}{2},
\end{align*}
from which we may infer that Majorana operators are self-inverse and satisfy the anti-commutation relations
\begin{align*}
    \{\gamma_j,\gamma_k\}=2\delta_{jk}\id.
\end{align*} 
    
    Given the block decomposition based on the Lieb-Robinson bound,
    we should implement each small block using $\polylog(Tn/\epsilon)$ 2-qubit gates.
    The Jordan-Wigner transformation, a representation of Majorana operators in terms of the Clifford algebra,
    is a first method one may consider:
    \begin{align}
    \gamma_{2j-1} &\mapsto \sigma^z_1 \otimes \cdots \otimes \sigma^z_{j-1} \otimes \sigma^x_j,\\
    \gamma_{2j}   &\mapsto \sigma^z_1 \otimes \cdots \otimes \sigma^z_{j-1} \otimes \sigma^y_j,
    \end{align}
    and the right-hand side is a tensor product of the $2\times 2$ Pauli matrices that satisfy
    \begin{align*}
	[\sigma_j,\sigma_k]=i2\epsilon_{jkl}\sigma_l,\quad \{\sigma_j, \sigma_k\}=\delta_{jk}\id,
    \end{align*}
    where $\epsilon_{jkl}$ is the anti-symmetric Levi-Civita symbol.
    Often, the tensor factor of $\sigma^z$ preceding $\sigma^x$ or $\sigma^y$ is called 
    a {\it Jordan-Wigner string}.
    In one spatial dimension, the above representation 
    where the ordering of $\gamma$ is the same as the chain's direction
    gives a local qubit Hamiltonian, since in any term Jordan-Wigner strings cancel.
    The ordering of fermions is thus very important.
    (Under periodic boundary conditions, at most one block may be nonlocal; 
    however, we can circumvent the problem 
    by regarding the periodic chain, a circle, 
    as a double line of finite length whose end points are glued:
    $ \big( [-1,+1] \times \{\uparrow,\downarrow\} \big) / \{ (-1,\uparrow) \equiv (-1,\downarrow), (+1,\uparrow) \equiv (+1,\downarrow)\} $.
    This trick doubles the density of qubits in the system,
    but is applicable in any dimensions for periodic boundary conditions.)
    
    In higher dimensions with fermions,
    a naive ordering of fermion operators
    turns most of the small blocks into nonlocal operators
    under the Jordan-Wigner transformation.
    However, fortunately, there is a way to get around this, at a modest cost,
    by introducing auxiliary fermions and 
    letting them mediate the interaction of a target Hamiltonian~\cite{Verstraete2005}.
    The auxiliary fermions are ``frozen,'' during the entire simulation, 
    by an auxiliary Hamiltonian that commutes with the target Hamiltonian.
    With a specific ordering of the fermions,
    one can represent all the new interaction terms as local qubit operators.
    The key is that if $c_j c_k$ is a fermion coupling whose Jordan-Wigner strings do not cancel,
    we instead simulate $c_j c_k \gamma_1 \gamma_2$, where $\gamma_{1,2}$ are auxiliary,
    such that the Jordan-Wigner strings of $\gamma_1,\gamma_2$ cancel those of $c_j,c_k$, respectively.
    The auxiliary $\gamma$'s may be ``reused'' for other interaction terms 
    if the interaction term involves fermions that are close in a given ordering of fermions.
    (Ref.~\cite{Verstraete2005} explains this manipulation for quadratic Hamiltonian terms
    but it is straightforward that any higher order terms can be treated similarly.
    They also manipulate the Hamiltonian for the auxiliary $\gamma$ to make it local after Jordan-Wigner transformation,
    but for our simulation purposes it suffices to initialize the corresponding auxiliary qubits.)
    In this approach, if we insert $\calO(1)$ auxiliary fermions per fermion in the target system,
    the gate and depth complexity is the same as if there were no fermions.
    Note that we can make the density of auxiliary fermions arbitrarily small
    by increasing the simulation complexity as follows.
    Divide the system with non-overlapping blocks of diameter $\ell$, which is e.g. $\calO(\log n)$,
    that form a hypercubic lattice.
    (These blocks have nothing to do with our decomposition by Lieb-Robinson bounds.)
    Put $\calO(1)$ auxiliary fermions per block,
    and order all the fermions lexicographically so that all the fermions in a block
    be within a consecutive segment of length $\calO(\ell^D)$ in the ordering.
    Interaction terms within a block have Jordan-Wigner string of length at most $\calO(\ell^D)$,
    and so do the inter-block terms using the prescription of \cite{Verstraete2005}.
    The gate and depth complexity of this modified approach has $\mathrm{poly}(\ell)$ overhead.

\appendix

\section{Heisenberg model benchmark}
\label{app:benchmark}
We may gain some intuition in applying Lieb-Robinson bounds to quantum simulation with a concrete example. The Heisenberg model offers one useful benchmark~\cite{Childs2017} for the performance of various quantum simulation algorithms. In the case of $1$D nearest-neighbor interactions with open boundary conditions, and where each spin is subject to an inhomogeneous magnetic field, its Hamiltonian is
\begin{align}
H=\sum_{j=0}^{n-1} \underbrace{(X_j X_{j+1} + Y_j Y_{j+1} + Z_j Z_{j+1} + z_j Z_j)}_{h_j},
\label{eq:Hamiltonian}
\end{align}
where $\{X,Y,Z\}$ are the single-qubit Pauli operators. 
Though this Hamiltonian may be diagonalized in a closed form  without the field term ($z_j=0$), 
in the presence of non-uniform $z_j$ this model can only be treated numerically in general.

To recap, there are two sources of error in the entire quantum simulation algorithm for implementing time-evolution $e^{-iTH}$.
One is from the decomposition of the full time-evolution operator $e^{-iTH}$ using Lieb-Robinson bounds
into $m = \calO(T n / \ell)$ blocks, 
and is bounded from above by $m \epsilon_\text{LR} = \calO( m e^{-\mu \ell})$ for some $\mu > 0$.
The other is from approximate simulations of the block unitary using known algorithms
 such as \cite{TS,LC16}.
If each block is simulated up to error $\epsilon_\square$,
then the total error $\epsilon$ of the algorithm is at most $m(\epsilon_{LR} + \epsilon_\square)$.
Thus, we need $\ell = \calO( \log(T n / \epsilon ))$. 

\begin{figure}
	\centering
	\includegraphics[width=0.75\textwidth]{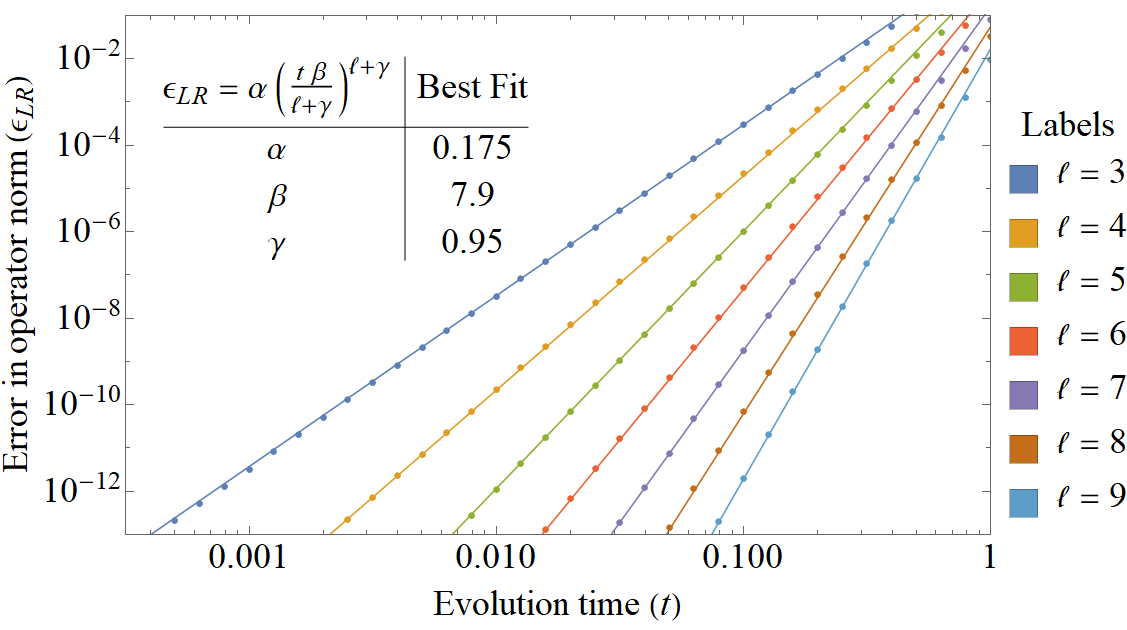}
	\caption{Numerical test of $m=1$ decomposition of the real-time evolution operator
		based on Lieb-Robinson bounds.
		The Hamiltonian is the antiferromagnetic one-dimensional Heisenberg model up to $11$ spins in \cref{eq:Hamiltonian}
		with open boundary condition.
		The error of the decomposition in \cref{eq:stairDecomposition}
		is almost independent of the position $a$ of the overlap within the system and is exponentially small in the overlap size $\ell$. All lines plotted are the best-fit to $\epsilon_{LR}=\alpha\left(\frac{t\beta}{\ell+\gamma}\right)^{\ell +\gamma}$.
	}
	\label{fig:numerics}
\end{figure}
\begin{figure}
	\centering
	\includegraphics[width=0.75\textwidth]{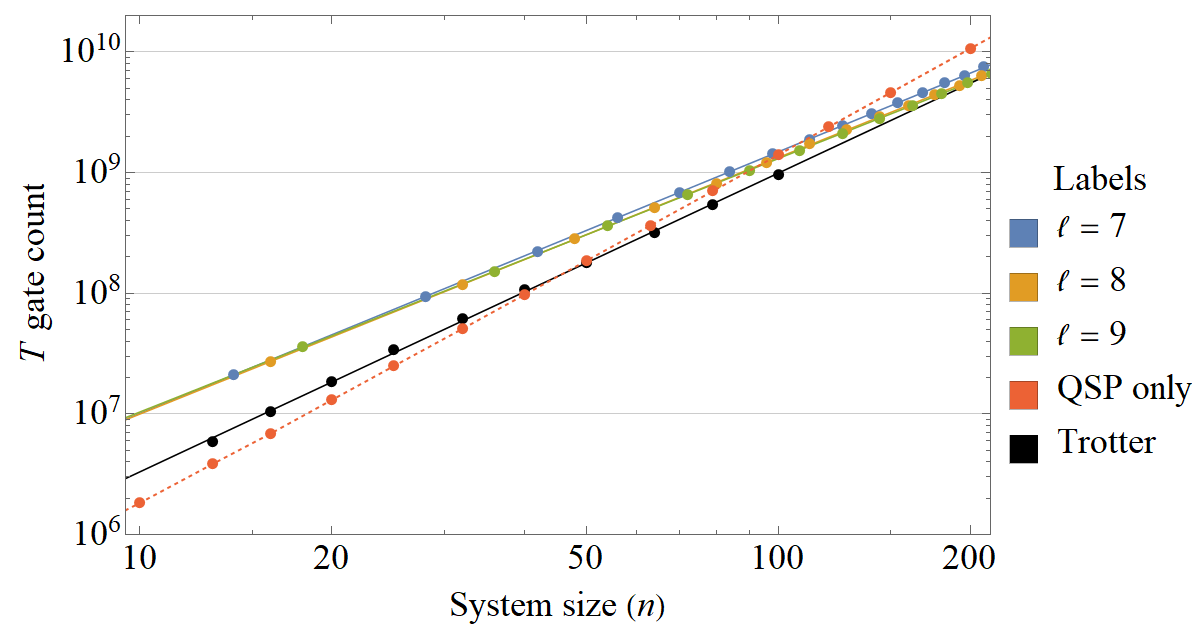}
	\caption{$T$ gate counts of simulating the Heisenberg model of~\cref{eq:Hamiltonian} for time of $n$, error $\epsilon=10^{-3}$, and $h_j\in[-1,1]$ chosen uniformly at random, using the Lieb-Robinson decomposition for overlap sizes of $\ell=7,8,9$. Plotted for reference is the complexity $\tilde{\calO}{(n^3)}$ of simulating the entire $n$-site system using Quantum Signal Processing (QSP)~\cite{QSP} without decomposing into blocks. Also plotted with data from~\cite{Childs2017} are the gate counts optimized over a numerical simulation of Lie-Trotter-Suzuki product formulas of various orders and step-sizes.
	}
	\label{fig:TGateCount}
\end{figure}

Before proceeding, we require estimates of the Lieb-Robinson contribution to error $\epsilon_\text{LR}$. By rescaling $H$ and $1/T$ by the same constant factor to ensure that $\|h_j\|\le 1$, this may be rigorously upper-bounded through \Cref{lem:LRB}. However,  it is also reasonable to numerically obtain better constant factors in the scaling of $\epsilon_\text{LR}$. Though simulation of the entire system of size $n$ is classically intractable, $\epsilon_\text{LR}$ can be obtained by classically simulating small blocks of size $\calO(\log{(n)})$, which is within the realm of feasibility. The decomposition ($m=1$) is 
\begin{align}
\exp( -\i t H ) 
    \simeq 
\exp\left( -\i t \sum_{j<b} h_j \right)
\exp\left( +\i t \sum_{j=a}^{b-1} h_j \right)
\exp\left( -\i t \sum_{j \ge a} h_j \right)
\label{eq:stairDecomposition}
\end{align}
so there are $\ell = b-a+1$ spins in the overlap. We computed the error for a wide range of $t$ up to $\ell=9$, and observed that the error is almost independent of the position of the overlap, 
and is also exponentially small in $\ell$. Note that the best fit $\epsilon_{LR}=\alpha\left(\frac{t\beta}{\ell+\gamma}\right)^{\ell +\gamma}$ in \Cref{fig:numerics} may be solved for $\ell = \calO{(t+\log{(1/\epsilon_{LR})})}$, and is consistent with \Cref{lem:decomposition}. 

Using the recursive decomposition into blocks shown in \Cref{fig:algorithm}, we now simulate $m/2$ blocks of size $\ell$ and $m/2$ blocks of size $2\ell$, both for time $t$, and each with error $\epsilon_\square = \frac{\epsilon}{3m}$. Holding $\ell$ constant, we may use fit approximation of $\epsilon_{LR}$ to simultaneously solve for the number of blocks $m=\frac{2T n}{t \ell}$ and the evolution time of each block $t$ such that the Lieb-Robinson error contribution $\epsilon_{LR} = \frac{\epsilon}{3m}$. Note that we may also invert the ordering of sequential stacks in \cref{eq:stairDecomposition} to merge blocks of size $2\ell$. For instance,
\begin{align}\nonumber
e^{ -\i t H} e^{ -\i t H }
&\simeq 
e^{\left( -\i t \sum_{j<b} h_j \right)}
e^{\left( +\i t \sum_{j=a}^{b-1} h_j \right)}
e^{\left( -\i t \sum_{j \ge a} h_j \right)}
e^{\left( -\i t \sum_{j \ge a} h_j \right)}
e^{\left( +\i t \sum_{j=a}^{b-1} h_j \right)}
e^{\left( -\i t \sum_{j<b} h_j \right)}
\\ 
&=e^{\left( -\i t \sum_{j<b} h_j \right)}
e^{\left( +\i t \sum_{j=a}^{b-1} h_j \right)}
e^{\left( -\i 2 t \sum_{j \ge a} h_j \right)}
e^{\left( +\i t \sum_{j=a}^{b-1} h_j \right)}
e^{\left( -\i t \sum_{j<b} h_j \right)}.
\label{eq:stairStackDecomposition}
\end{align}
Excluding boundary cases, this leads to fewer blocks $m=\frac{3T n}{2t \ell}$. Specifically, we may alternatively simulate $2m/3$ blocks of size $\ell$ for time $t$ and $m/3$ blocks of size $2\ell$ for time $2t$, and each with error $\epsilon_\square = \frac{\epsilon}{3m}$. Depending on the simulation algorithm used for each block, this may be slightly more efficient.

Similar to the benchmark in~\cite{Childs2017}, we obtain explicit gate counts in the Clifford+$T$ basis in~\Cref{fig:TGateCount} for simulating $e^{-iTH}$ with $T=n$, error $\epsilon=10^{-3}$, and $h_j\in[-1,1]$ chosen uniformly at random. We implement each block with the combined quantum signal processing~\cite{QSP} and qubitization~\cite{LC16} simulation algorithm. An outline of this algorithm together with certain minor circuit optimizations is discussed in \Cref{app:qspalg}.
Furthermore, the remaining error budget of $\epsilon/3$ is reserved for approximating arbitrary single-qubit rotations in the algorithm with Clifford+$T$ gates~\cite{Kliuchnikov2013Synthesis}.

\section{Further algorithmic improvements}
\label{app:extensions}

\subsection{Inhomogeneous interaction strength}

We can adapt the decomposition of the time-evolution unitary based on Lieb-Robinson bounds
when there is inhomogeneity in interaction strength across the lattice.
For this section, we do not assume that $\norm{h_X} \le 1$ for all $X \subseteq \Lambda$.
Instead, suppose there is one term $h_{X_0}$ in the Hamiltonian with $\norm{h_{X_0}} = J \gg 1$
while all the other terms $h_X$ have $\norm{h_X} \le 1$,
the prescription above says that we would have to divide the time step in $\lceil J \rceil $ pieces,
and simulate each time slice.
However, more careful inspection in the algorithm analysis
tells us that one does not have to subdivide the time step for the entire system.
For clarity in presentation,
let us focus on a one-dimensional chain where the strong term $h_{X_0}$ is at the middle of the chain.
We then introduce a cut as in \Cref{fig:algorithm} (a) at $h_{X_0}$.
The purpose is to put the strong term into $H_\text{bd}$ so that
the truncation error in \cref{eq:truncation} is manifestly at most linear in $J$.
Since the truncation error is exponential in $\ell$, the factor of $J$ in the error
can be suppressed by increasing $\ell$ by $\calO(\log J)$.
After we confine the strong term in a block of size $2\ell_0 = \calO( \log (J L T/\epsilon))$ in the middle of the chain,
the rest of the blocks can be chosen to have size $\calO(\log(LT/\epsilon))$ and do not have any strong term, 
and hence the time step $t$ can be as large as it would have been without the strong term.
For the block with Hamiltonian $H_\square$ that contains the strong term $h_{X_0}$, 
we can simply simulate $H_\square / J$ for time $Jt$.

\subsection{Reducing number of layers in higher dimensions}

Although we have treated the spatial dimension $D$ as a constant,
the number of layers for unit time evolution is $3^D$,
which grows rather quickly in $D$,
if we used the hyperplane decomposition as above.
We can reduce this number by considering a different tessellation of the lattice.

To be concrete, let us explain the idea in two dimensions.
Imagine a tiling of the two-dimensional plane
using hexagons of diameter, say, $10\ell$.
It is important that this hexagonal tiling is 3-colorable;
one can assign one of three colors, red, green, and blue,
to each of hexagons 
such that no two neighboring hexagons have the same color.

Let $R,G,B$ be the unions of red, green, and blue hexagons, respectively.
Each of $R,G,B$ consists of well separated hexagons.
Here, being well separated means that for any color $c$,
the $\ell$-neighborhood of a cell of color $c$
does not intersect with any other cell of color $c$.
Suppose we had implemented the time evolution 
$U({R\cup G})$ for the Hamiltonian $H_{R \cup G}$.
Consider $\ell$-neighborhood $B^+$ of $B$.
$B^+ \cap (R \cup G)$ consists of separated ``rings'' of radius $\simeq 6\ell$
and thickness $\simeq \ell$.
(See \cref{fig:hexagonaltiling}.)
We can now apply \Cref{lem:decomposition} to $R\cup G$ and $B^+$
to complete the unit time evolution for the entire system $R \cup G \cup B$.
The unitaries needed in addition to $U(R \cup G)$ are 
the backward time-evolution operator on $(R \cup G) \cap B^+$,
which is a collection of disjoint unitaries on the rings,
and the forward time-evolution operator on $B^+$,
which is a collection of disjoint unitaries on enlarged hexagons.

\begin{figure}
\centering
\includegraphics[width=0.7\textwidth, trim ={0ex 50ex 120ex 0ex}, clip]{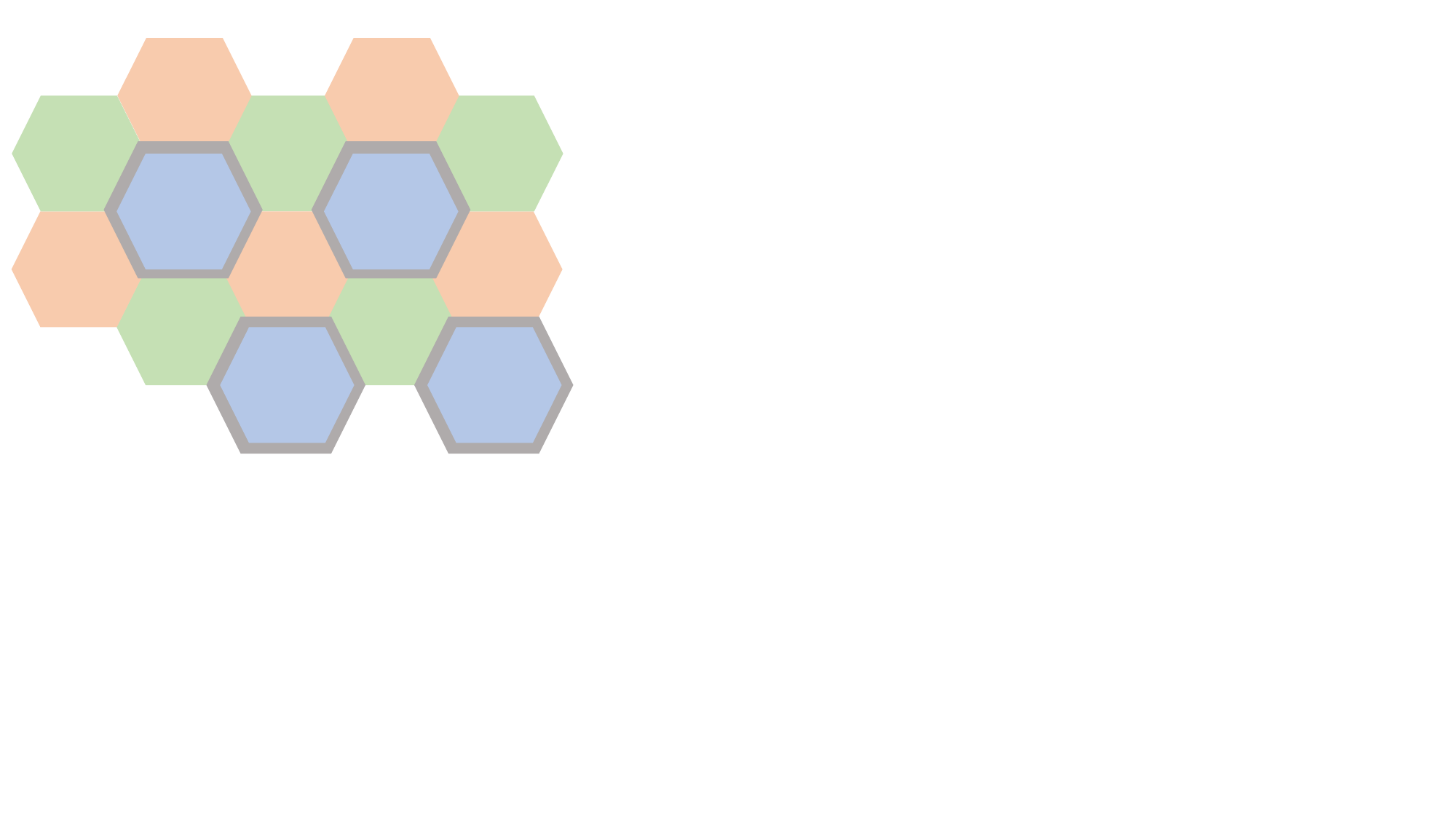}
\caption{Hexagonal tiling of 2D plane with 3 colors.
The grey regions (``rings'') that encircle blue hexagons 
represent the intersection of the $\ell$-neighborhood $B^+$ of $B$ and $R\cup G$.}
\label{fig:hexagonaltiling}
\end{figure}

The time evolution $U(R \cup G)$ is constructed in a similar way.
We consider an $\ell$-neighborhood of $G$ within $R \cup G$.
The enlarged part $G^+ \cap R$ consists of line segments of thickness 
$\ell$, and it is clear that $R \setminus G^+$ is $\ell$-away
from $G$. We can again apply \Cref{lem:decomposition}.

In summary, the algorithm under the 3-colored tesellation is
(i) forward-evolve the blocks in $R$,
(ii) backward-evolve the blocks in $R \cap G^+$,
(iii) forward-evolve the blocks in $G^+ \cap (R \cup G)$,
(iv) backward-evolve the blocks in $(R\cup G)\cap B^+$,
and
(v) forward-evolve the blocks in $B^+$.

In general, if the layout of qubits allows
$\alpha$-colorable tessellation, where $\alpha = 2,3,4,\ldots$,
such that the cells of the same color are well separated,
then we can decompose unit time evolution into $2\alpha -1$ layers
by considering fattened cells of the tessellation.
The proof is by induction. When $\alpha =2$, it is clear.
For larger $\alpha$, we implement the forward-evolution
on the union $A$ of $\alpha -1$ colors using $2\alpha -3$ layers
by the induction hypothesis,
and finish the evolution by backward-evolution on $A \cap B^+$,
where $B$ is the union of the last color and $B^+$ is the 
$\ell$-neighborhood of $B$,
and then forward-evolution on $B^+$.
This results in $2\alpha -1 $ layers in total.

A regular $D$-dimensional lattice can be covered 
with $D+1$ colorable tessellation.
One such coloring scheme is obtained by any triangulation of $\mathbb{R}^D$,
and coloring each 0-cell by color ``0'',
and each 1-cell by color ``1'', and so on,
and finally fattening them.
This is a well known fact; 
see for example \cite{BravyiPoulinTerhal2010Tradeoffs}.

For three dimensions, there exists a more ``uniform'' 4-colorable tessellation.
Consider the body-centered cubic (BCC) lattice,
spanned by basis vectors $(2,0,0)$, $(0,2,0)$, and $(1,1,1)$.
Color each BCC lattice point $p=(x,y,z)$
by the rule $c = x+y+z \mod 4$.
The Voronoi tessellation associated with this colored BCC lattice
is a valid 4-colored tessellation for our purpose,
as seen as follows.
For $c=0,1,2,3$, the sublattice of color $c$ 
is spanned by $(2-c,-2-c,-c), (4-c,4-c,4-c), (-c,-c,4-c)$.
It is easy to check that the shortest distance 
between two distinct lattice points of the same color is $2\sqrt 2 \simeq 2.828$,
by, for example, an exhaustive search in the box $\{-5,\ldots,5\}^3 \subset \ZZ^3$. 
By definition of Voronoi tessellation,
a point $q \in \RR^3$ belongs to a cell $C$ of color $c$
if and only if the closest lattice point to $q$ is unique and has color $c$.
Since $C$ is convex, a point $q'$ on the boundary of $C$ 
is the furthest from the central lattice point of $C$
if $q'$ is equidistant from the four sublattices.
Indeed, $(1,0,\frac 1 2)$ is equidistant by distance $\frac 1 2 \sqrt 5 \simeq 1.118$ 
from distinctly colored points $(0,0,0),(1,-1,1),(2,0,0),(1,1,1)$,
which define a tetrahedron inside which there is no other lattice point.
Hence, each cell is contained in a ball of radius $\frac 1 2 \sqrt 5$,
which is smaller than $\sqrt 2$, the half of the minimum distance between distinct points of the same color.
Therefore, the cells of the same color are separated.

\section{Lieb-Robinson Bounds with Bounded Commutators}
\label{app:clr}

\subsection{Introduction and Assumptions}

One advantage of the method described in this paper is that if the Lieb-Robinson velocity is small, then the method becomes more accurate.  One case in which this occurs is if the terms in the Hamiltonian have a small commutator with each other.  This section considers the Lieb-Robinson velocity under such a small commutator assumption.
Note that Trotter-Suzuki methods also improve in accuracy if the Hamiltonian terms have a small commutator~\cite{Childs2017}, but many other time-simulation methods do not.

Bounds on the Lieb-Robinson velocity with bounded commutators 
were first considered in Ref.~\cite{premont2010lieb}.  
Our results generalize their work in several ways.  
The previous work considered strictly local Hamiltonians with a bound on the commutator of two terms 
(but without any bound on the norm of a term),
but it was under the further assumption that the Hamiltonian 
was a sum of two different strictly local Hamiltonians $H=H_0+H_1$, 
such that $H_0$ and $H_1$ were each a sum of exactly commuting terms,
with the commutator bound applying to the commutator of a term in $H_0$ with a term in $H_1$.
We consider strictly local interactions with a bound on the commutator 
without any bound on the norm of a term, 
but without requiring the further assumption that the Hamiltonian can be decomposed 
as a sum of two exactly commuting Hamiltonians.
Also, we do not require any bound on higher order commutators as was used in \cite{premont2010lieb};
however, we do find in some cases tighter bounds when we assume such higher order bounds.

Additionally, we consider exponentially decaying interactions, 
rather than strictly local interactions.  
In this case, instead of giving a bound on the commutator of two terms, 
we find it necessary to give a bound on the norm of terms as well as the bound commutator.
This is required to express the exponential decay appropriately.
The exponential decay and commutator bounds that we consider are as follows:
consider a Hamiltonian $H=\sum_X h_X$, 
where each $X$ represents some set of sites and each $h_X$ a Hermitian operator supported on $X$.
The sum is over all possible subsets of the lattice.
The terms $h_X$ obey a commutator condition that there exists $0 \le \eta \le 1$ such that
for all $X$ and $Y$
\begin{align}
\label{commcond}
\Vert [h_X,h_Y] \Vert \leq 2 \eta \norm{h_X} \norm{h_Y}.
\end{align}
The exponential decay is imposed as follows.
Introduce a metric $\dist(x,y)$ between pairs of sites~$x,y$.
As before, for any two sets $X,Y$ of sites, we write $\dist(X,Y) = \min_{x \in X, y \in Y} \dist(x,y)$.
For any set $X$ of sites, its diameter $\diam(X)$ is defined to be $\max_{x,y \in X} \dist(x,y)$.
Assume that there are constants $\zeta, \mu > 0$ such that for any site $x$,
\begin{align}
\label{locassum}
\sum_{X\ni x} \norm{h_X} |X|^2 \exp(\mu \,\diam(X)) \leq \zeta < \infty
\end{align}
where $|X|$ denotes the cardinality of set $X$.
The assumption \eqref{locassum} will not be used until \cref{lrv};
we will indicate where it is used as many of the early bounds do not use this assumption and only use \cref{commcond}.
This assumption is slightly stronger than previous exponential decay assumptions 
such as in Ref.~\cite{HastingsKoma2006} as we have $|X|^2$ rather than $|X|$.
The reason for this will be clear later.
Note that we do not impose $\norm{h_X} \le 1$ in this appendix;
the strength of interaction is bounded only through \cref{locassum}.

Throughout, we write
\begin{align*}
O(t;J) := \exp(\i J t) O \exp(-\i J t)
\end{align*}
for any operator $O$ and any hermitian operator $J$.
If $J$ is the full Hamiltonian $H$ of the system,
we omit $H$ and write 
\begin{align*}
O(t) = O(t;H).
\end{align*}

Most of the appendix is devoted to the exponential decay case.
In \cref{sl}, we consider the strictly local case.
The main result that we will prove under the exponential decay assumptions is:
\begin{lemma}
\label{clrlemma}
Assume that assumptions (\ref{commcond},\ref{locassum}) hold.
Then, for any operator $A$ supported on a set $X$ and operator $B$ supported on a set $Y$ we have
\begin{align}
\Vert [A(t),B] \Vert \leq \frac{2}{\sqrt{\eta}}\norm{A} \cdot \norm{B} 
\left( \exp\left( \zeta |t| \sqrt{8\eta} \right) - 1 \right)
\sum_{x\in X} \exp(-\mu\,\dist(x,Y)).
\end{align}
\end{lemma}
Let $v_{LR}$ be chosen greater than
$\zeta \sqrt{8\eta} / \mu$.
Then, for large $t$, at fixed $\dist(X,Y)/t \leq v_{LR}$, 
for bounded $|X|$, the above bound tends to zero,
giving a Lieb-Robinson velocity proportional to $\sqrt{\eta}$.

A Lieb-Robinson velocity proportional to $\sqrt{\eta}$ might initially be surprising: one might hope to have a bound proportional to $\eta$.  However, one can see that this is the best possible under these assumptions.  Consider any local Hamiltonian $H=\sum_X h^0_X$ (without a commutator condition) with $\Vert h^0_X\Vert$ of order unity and Lieb-Robinson velocity $v^0_{LR}$ also of order unity.  Now, consider a new Hamiltonian $H=\sum_X h_X$ with $h_X=1+\sqrt{\eta} h^0_X$; here $1$ simply denotes the identity operator.  Then, the commutator of any two terms is proportional to $\eta$, while the norm of the terms is still of order $1$ and the Lieb-Robinson velocity is proportional to $\sqrt{\eta}$.  If the reader does not like adding the identity to define $h_X$ as a way of ensuring $\norm{h_X} \sim 1$, one could instead add some other exactly commuting terms of norm $1$ which act on some additional degrees of freedom.

In the proof below, we use notations as if the Hamiltonian was time-independent.
However, \Cref{clrlemma} is valid (replacing the exponential with a time-ordered exponential) if the Hamiltonian is a piecewise continuous function of time,
provided that assumptions (\ref{commcond},\ref{locassum}) hold for all time.

\subsection{Bound on Commutator Assuming \cref{commcond}}
\label{bc}

We wish to bound
\begin{align}
\left\Vert \left[\exp(\i H t) A_X \exp(-\i H t),B_Y	\right] \right\Vert.
\end{align}
where $A_X$ is supported on $X$ and $B_Y$ is supported on $Y$.
In what follows, we will drop the subscripts $X,Y$ from $A_X$ and $B_Y$ to avoid overly cluttering the notation.

Define
\begin{align}
C_B(X,t) &= \sup_{A \in {\cal A}_X : \norm{A} \le 1} \Vert [A(t),B] \Vert,\\
D_B(X,t) &= \Vert [h_X(t),B] \Vert
\end{align}
where ${\cal A}_X$ is the algebra of operators supported on $X$.
For any two sets $X$ and $Y$ of sites we will write $X \sim Y$ in place of $X \cap Y \neq \emptyset$.
Also, the notation $Z_1 \sim Z_2 \sim \cdots \sim Z_m$ will mean that 
$Z_j \sim Z_{j+1}$ for every $j = 1,\ldots, m-1$.
This does \emph{not} necessarily mean that $Z_j \sim Z_k$ for $j + 1 < k$.
For any set $X$, define $I_X$ as
\begin{align}
I_X=\sum_{Z: Z \sim X} h_Z .
\end{align}

Let $\dt$ be a small positive real number.  We consider a finite system so that $H$ has finite operator norm.
The quantities $\odts$ and $\calO(1/n)$ in \cref{lreq,firsteq,o1n}
will have a hidden dependence on system size, but after taking a limit ${\rm d}t \rightarrow 0$ (equivalent to $n\rightarrow \infty$) the bounds in the subsequent equations will not depend on system size and as a result the Lieb-Robinson velocity will not depend on system size.

We have
\begin{align}\label{lreq}
&
\Vert [A(t+\dt),B] \Vert
\\ \nonumber 
&= 
\Vert [A(\dt),B(-t)] \Vert \\ \nonumber
&= 
\Vert[\exp(\i I_X \dt) A \exp(-\i I_X \dt),B(-t)] \Vert + \odts \\ \nonumber
&=
\Vert[A,\exp(-\i I_X \dt) B(-t) \exp(\i I_X\dt)]\Vert+\odts\\ \nonumber
&=
\Vert[A,B(-t)-\i\dt [I_X,B(-t)]\Vert+\odts \\ \nonumber
&\leq \norm{[A,B(-t)]} + 2 \dt \norm{A} \cdot \norm{[I_X,B(-t)]} +\odts\\ \nonumber
&\leq \norm{[A,B(-t)]} + 2 \dt \norm{A} \sum_{Z: Z \sim X} \norm{[h_Z,B(-t)]} + \odts
\end{align}
By definitions of $C_B(X,t)$ and $D_B(X,t)$, it follows that
\begin{align}
\label{firsteq}
C_B(X,t+\dt)\leq C_B(X,t)+ 2 \dt \sum_{Z: Z \sim X} D_B(Z,t)+\odts.
\end{align}
Hence,
for any positive integer $n$,
\begin{align}
\label{o1n}
C_B(X,t)\leq C_B(X,0) + \frac{1}{n} \sum_{j=0}^n \sum_{Z: Z \sim X} 2D_B(Z,tj/n) + \calO(1/n).
\end{align}
For finite operator norm of $H$, the above expression converges to an integral as $n\rightarrow \infty$, so
\begin{align}
\label{int1}
C_B(X,t) \leq C_B(X,0)+\sum_{Z: Z \sim X} 2 \int_0^{|t|} D_B(Z,s) \rd s.
\end{align}
Also, we have
\begin{align}
&
\Vert [h_X(t+\dt),B] \Vert
\\ \nonumber 
&= 
\Vert [h_X(\dt),B(-t)] \Vert \\ \nonumber
&=
\Vert [h_X,B(-t)] + \i \dt [[I_X,h_X],B(-t)] \Vert + \odts \\ \nonumber
&\leq 
\Vert [h_X,B(-t)] \Vert + \dt \sum_{Z: Z\sim X} \Vert [[h_Z,h_X],B(-t)] \Vert.
\end{align}
By definitions of $C_B(X,t)$ and $D_B(X,t)$ and assumption (\ref{commcond}), it follows that
\begin{align}
D_B(X,t+\dt) \leq D_B(X,t)+ 2 \dt \sum_{Z: Z\sim X} \eta  \norm{h_Z} \cdot \norm{h_X} C_B(Z \cup X,t).
\end{align}
Hence,
\begin{align}
\label{int2}
D_B(X,t) 
\leq 
D_B(X,0) +  
\sum_{Z: Z\sim X} 2\eta \int_0^{|t|} \norm{h_Z} \cdot \norm{h_X} C_B(Z \cup X,s) {\rm d}s.
\end{align}
Note that for an arbitrary set $Z$ of sites and any $t$
\begin{align}
	C_B(Z,0) &\le 2 \norm{B} \delta_{Z \sim Y},           & C_B(Z,t) &\le 2 \norm{B}, \label{eq:CBZ0} \\
	D_B(Z,0) &\le 2 \norm{B} \norm{h_Z} \delta_{Z \sim Y}, & D_B(Z,t) &\le 2 \norm{B} \norm{h_Z} \nonumber
\end{align}
where $\delta_{\mathbb P}=1$ if the predicate $\mathbb P$ is true and $\delta_{\mathbb P} = 0$ otherwise.
Thus we may rewrite \cref{int1,int2} as
\begin{align}
C_B(X,t) &\leq 
2\norm{B} \delta_{X \sim Y} + \sum_{Z: X \sim Z} 2 \int_0^{|t|}\rd s\, D_B(Z,s), \label{recurCB}\\
D_B(X,t) &\leq
2 \norm{B} \norm{h_X} \delta_{X \sim Y}
+
\sum_{Z: X \sim Z} 2\eta \norm{h_X} \cdot \norm{h_Z} \int_0^{|t|} \rd s\, C_B(X \cup Z,s) \delta_{X \not\sim Y} . \label{recurDB}
\end{align}
Since $X$ and $Y$ are arbitrary, we may use \cref{recurDB} in \cref{recurCB}; 
for any sets $Z_{j-2},Z_{j-1}$ we have:
\begin{align}
&\frac{C_B(Z_{j-2} \cup Z_{j-1}, s_{j-1})}{2\norm{B}} \delta_{Z_{j-2} \not\sim Y} \nonumber \\
&\le \delta_{Z_{j-1} \sim Y} \label{eq:recursion} \\
&\quad +
(2 s_{j-1}) \sum_{\begin{subarray}{c} Z_j :\\ (Z_{j-2} \cup Z_{j-1}) \sim Z_j \sim Y \end{subarray}} \norm{h_{Z_j}} \nonumber \\
&\quad +
\eta \cdot 2^2 \int_0^{s_{j-1}} \rd s_j \int_0^{s_j} \rd s_{j+1} 
\sum_{\begin{subarray}{c}Z_j,Z_{j+1}:\\ (Z_{j-2}\cup Z_{j-1}) \sim Z_j \sim Z_{j+1}\end{subarray}}
 \norm{h_{Z_j}}\norm{h_{Z_{j+1}}}
\frac{C_B( Z_j \cup Z_{j+1}, s_{j+1} )}{2\norm{B}}\delta_{Z_j \not\sim Y}  \nonumber
\end{align}
where $s_{j-1} \ge 0$.
This is our fundamental recursive inequality.
Note that we have only the constraint that $Z_{j-1} \sim Y$ in the second line
rather than the slightly looser constraint that $(Z_{j-2} \cup Z_{j-1}) \sim Y$.
This is due to the constraint $\delta_{Z_{j-2} \not\sim Y}$ in the first line.
Assuming $X \cap Y=\emptyset$ and  setting $Z_{-1} = \emptyset, Z_0 = X$, we have
\begin{align}
\label{sumbound}
\frac{C_B(X,t)}{2\norm{B}} 
&\leq
 (2|t|)  \sum_{Z_1: X \sim Z_1 \sim Y} \norm{h_{Z_1}}
\nonumber \\ 
&\quad + \eta \frac{(2|t|)^2}{2!} \sum_{Z_1, Z_2: X \sim Z_1 \sim Z_2 \sim Y} \norm{h_{Z_1}} \norm{h_{Z_2}} \\ \nonumber
&\quad + \eta \frac{(2|t|)^3}{3!} \sum_{Z_1, Z_2:X \sim Z_1 \sim Z_2} \norm{h_{Z_1}} \norm{h_{Z_2}} \sum_{Z_3: (Z_1 \cup Z_2) \sim Z_3 \sim Y} \norm{h_{Z_3}}
\\ \nonumber
&\quad + \eta^2 \frac{(2|t|)^4}{4!} \sum_{Z_1,Z_2:X \sim Z_1 \sim Z_2} \norm{h_{Z_1}}\norm{h_{Z_2}} \sum_{Z_3, Z_4: (Z_1 \cup Z_2)\sim Z_3 \sim Z_4 \sim Y} \norm{h_{Z_3}} \norm{h_{Z_4}}
\\ \nonumber
&\quad +\cdots 
\end{align}
That is,
\begin{align*}
\frac{C_B(X,t)}{2\norm{B}}  \leq \sum_{k\geq 1}\frac{(2|t|)^k}{k!} \eta^{\lfloor k/2 \rfloor} \sum_{Z_1,Z_2,\ldots,Z_k}\,
\Bigl(\prod_{j=1}^k \norm{h_{Z_j}} \Bigr) \Bigl(\prod_{j \, {\rm odd}} \delta_{(Z_{j-2} \cup Z_{j-1}) \sim Z_j}\Bigr)
\Bigl(\prod_{j \, {\rm even}} \delta_{Z_{j-1} \sim Z_j}\Bigr)
\delta_{Z_k \sim Y},
\end{align*}
where the notation $\delta_{S\sim T}$ for 
two sets $S,T$ is an indicator function that is $1$ if $S \cap T \neq \emptyset$ and $0$ otherwise.

\subsection{Lieb-Robinson Velocity}
\label{lrv}

We now use assumption \eqref{locassum}, especially the following consequences.
\begin{proposition}\label{prop:setsums}
For arbitrary sets $P,Q,R,S$ of sites
\begin{align}	
\sum_{Q: P \sim Q \sim S} \norm{h_Q} &\le \zeta \sum_{p \in P} e^{ - \mu \, \dist(p,S) }, \label{eq:middleSum}\\
\sum_{Q: P \sim Q } \norm{h_Q} |Q|^2 e^{-\mu\,\dist(Q,S)} & \le \zeta \sum_{p \in P} e^{- \mu \, \dist(p,S) }, \label{eq:rightSum}\\
\sum_{Q,R: P \sim Q \sim R \sim S} \norm{h_Q} \norm{h_R} &\le \zeta^2 \sum_{p \in P} e^{-\mu\,\dist(p,S)}, \label{eq:middle2Sum}\\
\sum_{Q,R: P \sim Q \sim R} |Q \cup R| \norm{h_Q} \norm{h_R} e^{-\mu\,\dist(Q\cup R, S)} &\le 2\zeta^2 \sum_{p\in P} e^{-\mu\,\dist(p,S)}. \label{eq:rightBranchSum}
\end{align}
Here, $\sum_{p \in P} e^{-\mu\,\dist(p,S)} \le |P| e^{-\mu\,\dist(P,S)}$. 
\end{proposition}
\begin{proof}
	\eqref{eq:middleSum}: 
	Observe $\sum_{Q:P \sim Q \sim S} \le \sum_{p \in P} \sum_{Q:p \in Q \sim S}$
	whenever the summand is nonnegative.
	Using assumption \eqref{locassum} and $\diam(Q) \ge \dist(p,S)$ when $p \in Q \sim S$,
	we have an upper bound as
	$\sum_{p} \sum_{Q: p \in Q \sim S} \norm{h_Q} e^{\mu\,\diam(Q) - \mu\,\dist(p,S)} \le \zeta \sum_p e^{-\mu\,\dist(p,S)}$.
	
	\eqref{eq:rightSum}: 
	Similarly, we use that 
	$\diam(Q) + \dist(Q,S) \ge \dist(p,S)$ when $p \in Q$
	by triangle inequality of the metric.

	\eqref{eq:middle2Sum}:
	Use \eqref{eq:middleSum} for the sum over $R$ and then use \eqref{eq:rightSum} for the sum over $Q$.
	
	\eqref{eq:rightBranchSum}:
	Use $|Q \cup R| \le |Q| \cdot |R|$ since  $Q\sim R$.
	We then separately bound two cases where (i) $\dist(Q \cup R, S) = \dist(R,S)$ and (ii) $\dist(Q \cup R, S) = \dist(Q,S)$.
	The sum of $(i)$ and $(ii)$ is an upper bound to the original sum.
	For case (i), we use \eqref{eq:rightSum} for the sum over $R$ to have an upper bound 
	$\zeta \sum_{Q: P \sim Q} |Q|^2 \norm{h_Q} e^{-\mu\,\dist(Q,S)}$, 
	and then use \eqref{eq:rightSum} again for the sum over $Q$ to have an upper bound
	$\zeta^2 |P| e^{-\mu\,\dist(P,S)}$.
	For case (ii), we use either \eqref{eq:middleSum} or \eqref{eq:rightSum} to sum over $R$,
	which gives $\zeta\sum_{Q:P \sim Q} |Q|^2 \norm{h_Q} e^{-\mu\,\dist(Q,S)}$, and then use \eqref{eq:rightSum}.		
\end{proof}

Now, we can use \cref{eq:middleSum} for the innermost sum in any odd-$k$-th line of \cref{sumbound},
and \cref{eq:middle2Sum} for that in any even-$k$-th line.
Once the innermost sum is bounded, we use \cref{eq:rightBranchSum} for $\lfloor (k-1)/2 \rfloor$ times.
This is the point at which the dependence on $|X|^2$ in assumption \eqref{locassum} is necessary.
Therefore, the $k$-th line is bounded by
\begin{align}
\frac{(2|t|)^k}{k!} \eta^{\lfloor k/2\rfloor} 2^{\lfloor (k-1)/2\rfloor} \zeta^k \sum_{x\in X} \exp(-\mu \,\dist(x,Y)).
\end{align}
Summing over $k$ we find that
\begin{align}
\label{finalbound}
C_B(X,t) \leq \frac{2}{\sqrt{\eta}}\norm{B} 
\left( \exp\left( \zeta |t| \sqrt{8 \eta} \right) - 1 \right)
\sum_{x\in X} \exp(-\mu\,\dist(x,Y)),
\end{align}
so that \cref{clrlemma} follows.

\subsection{Proof of \Cref{lem:LRB}} \label{app:pfLRB}

A slight modification of the previous section proves \cref{lem:LRB}.
Recall that $H_\Omega = \sum_{Z : Z\subseteq \Omega} h_Z$,
and it suffices to assume $X \subset \Omega$.
Then,
\begin{align}
\norm{
A_X(t;H_\Omega) - A_X(t;H)
}
&=
\norm{
	\int_0^t \rd s \partial_s (U^{H_\Omega}_{s} U^{H}_{t-s})^\dagger A_X U^{H_\Omega}_{s} U^{H}_{t-s}
}\nonumber\\
&\le
\int_0^{|t|} \rd s
\norm{
(U^{H}_{t-s})^\dagger[H - H_\Omega, (U^{H_\Omega}_s)^\dagger A_X U^{H_\Omega}_s] U^{H}_{t-s}
}\label{eq:AtOmegaAt}\\
&\le
\int_0^{|t|}\rd s
\sum_{Y: Y \sim \Omega^c}
\norm{[A_X(s; H_\Omega),h_Y]} \nonumber
\end{align}
(In the last inequality the sum over $Y$ could be further restricted to those with $Y \sim \Omega$,
but the present bound will be enough.)
The last line can be bounded by multiplying \cref{sumbound}
by $2\norm{B} = 2 \norm{h_Y}$ and summing over $Y$ such that $Y\sim \Omega^c$.
The innermost sum of \cref{sumbound} is modified to
\begin{align}
&\sum_{Z_k: (Z_{k-2} \cup Z_{k-1}) \sim Z_k} \norm{h_{Z_k}} 
\sum_{Y: Z_k \sim Y \sim \Omega^c } \norm{h_Y} & \text{if $k$ is odd,}\\
&\sum_{Z_{k-1}, Z_k: (Z_{k-3} \cup Z_{k-2}) \sim Z_{k-1} \sim Z_k} \norm{h_{Z_{k-1}}}\norm{h_{Z_k}}
\sum_{Y: Z_k \sim Y \sim \Omega^c } \norm{h_Y} & \text{if $k$ is even.}\nonumber
\end{align}
The sum over $Y$ here is bounded by \cref{eq:middleSum},
and the remaining sum is bounded by applying \cref{eq:rightSum} once or twice.
The net effect of the modification is that there is an extra factor of $\zeta$,
and that the distance is now measured to $\Omega^c$ instead of $Y$ before.
We conclude that
\begin{align}
\norm{
	A_X(t;H) - A_X(t;H_\Omega)
}
&\le
\frac{2\zeta|t|}{\sqrt{\eta}} \norm{A}
\left( \exp\left( \zeta |t| \sqrt{8 \eta} \right) - 1 \right)
\sum_{x\in X} \exp(-\mu\,\dist(x,\Omega^c)).
\end{align}
Since $\eta \le 1$, this proves a variant of \cref{lem:LRB} 
where the locality assumption is given by \cref{locassum}.

In \cref{lem:LRB} we assumed strictly local interactions such that $h_X = 0$ if $\diam(X)> 1$
in a $D$-dimensional lattice.
This strict locality condition implies \cref{locassum} with arbitrary $\mu > 0$
($\zeta$ can be estimated as a function of $\mu$),
and hence the bound is sufficient for \cref{thm:main}.
However, one can prove a stronger bound for the strictly local interactions.
Since $D_B(X,t) \le \norm{h_X} C_B(X,t)$, assuming $X \not \sim Y$, we have
\begin{align}
C_B(X,t) 
&\le
C_B(X,0) + 2 \sum_{Z:X \sim Z} \norm{h_Z} \int_0^{|t|} \rd s C_B(Z,s)\\
&\le
\sum_{k=1}^{\ell-1} 
\frac{(2|t|)^k}{k!}  
\left(\sum_{Z_1,\ldots,Z_k: \text{linked}} \prod_{j=1}^k \norm{h_{Z_j}} \right) 
C_B (Z_k,0) 
\nonumber \\
&\qquad
+
2^\ell \int_{\Delta^\ell} \rd^\ell \mathbf t 
\left(\sum_{Z_1,\ldots,Z_\ell: \text{linked}} \prod_{j=1}^\ell \norm{h_{Z_j}} \right) 
C_B (Z_\ell,t_\ell)
\label{eq:brr}\\
&\le
2 \norm{B} \frac{(2|t|)^\ell}{\ell !}
\sum_{Z_1,\ldots,Z_\ell: \text{linked}} \prod_{j=1}^\ell \norm{h_{Z_j}} 
\quad (\ell = \lfloor \dist(X,Y) \rfloor) \label{eq:brrl}
\end{align}
where the factor $|t|^k/k!$ 
is $\int_{\Delta^k} \rd \mathbf t = \int_0^{|t|} \rd t_1 \int_0^{t_1} \rd t_2 \cdots \int_0^{t_{k-1}} \rd t_k$ 
(the volume of a $k$-dimensional simplex $\Delta^k$),
``linked'' means that $X \sim Z_1 \sim Z_2 \sim \cdots \sim Z_k$,
and $t_\ell$ is the last component of $\mathbf t$.
Here, \cref{eq:brr} is valid for any integer $\ell \ge 1$,
but in \cref{eq:brrl} we set $\ell$ to the specific value.
\cref{eq:brrl} follows by \cref{eq:CBZ0} 
and a trivial bound $C_B(Z_\ell,t_\ell) \le 2\norm{B}$;
\cref{eq:CBZ0} forces $Z_k \sim Y$, but due to the locality
there is no nonzero ``link'' from $X$ to $Y$ 
using $\ell -1$ segments or less,
so the first $\ell -1 $ terms of \cref{eq:brr} are zero.

Similarly to \Cref{prop:setsums},
we see for arbitrary sets $P,Q,S$ of sites
\begin{align}
\sum_{Q:P\sim Q \sim S } \norm{h_Q} &\le \zeta_0 \sum_{p \in P} \delta_{\dist(p,S) \le 1},\\
\sum_{Q:P\sim Q } |Q| \norm{h_Q} \delta_{\dist(Q,S) \le d} &\le \zeta_0 \sum_{p \in P} \delta_{\dist(p,S) \le d+1}
\end{align}
where $\zeta_0 = \max_{x \in \Lambda} \sum_{Q \ni x} |Q| \norm{h_Q}$.
Note that $\zeta_0$ is bounded by the number of nonzero Hamiltonian terms that may act on a site $p$
and the number of sites in a ball of diameter 1.
We conclude that
\begin{align}
C_B(X,t) 
&\le 
2 \norm{B} \sum_{x \in X} \frac{(2 \zeta_0 |t|)^\ell }{ \ell !}, 
\\
\norm{[A_X(t),B_Y]}
&\le
2 \norm{A} \norm{B} |X| \frac{(2\zeta_0|t|)^\ell}{\ell!} \text{ where } \ell = \lfloor \dist(X,Y) \rfloor .
\end{align}
By manipulation analogous to \cref{eq:AtOmegaAt}, we also have
\begin{align}
\norm{	A_X(t;H) - A_X(t;H_\Omega)	} \le |X| \norm{A_X} \frac{(2 \zeta_0 |t|)^{\ell}}{\ell!}
\text{ where } \ell = \lfloor \dist(X,\Lambda \setminus \Omega)\rfloor .
\end{align}
This completes the proof of \Cref{lem:LRB}.

\subsection{Higher Order Commutators}

Finally, let us remark that even better bounds can be proven 
if one assumes a bound on higher-order commutators.  
For example, if we assume that
\begin{align}
\Vert [[h_X,h_Y],h_Z] \Vert \leq \eta' \norm{h_X} \cdot \norm{h_Y} \cdot \norm{h_Z},
\label{highercomm}
\end{align}
we can prove a better bound for sufficiently small $\eta'$.
This is done by an straightforward generalization of the above results:
In addition to quantities $C_B(X,t)$ and $D_B(X,t)$, 
define also a quantity $E_B(X,Y,t)$ as
\begin{align}
E_B(X,Y,t)=\Vert [[h_X(t),h_Y(t)],B] \Vert.
\end{align}
Then, just as we bound $C_B(X,t+\dt)-C_B(X,t)$ in terms of $D_B(X,t)$ above, we also bound 
$D_B(X,t+\dt)-D_B(X,t)$ in terms of $E_B(X,Y,t)$ summed over sets $Y$ that intersect $X$.
Then, we bound $E_B(X,Y,t+\dt)-E_B(X,Y,t)$ using \cref{highercomm}.
Extensions to even higher commutators follow similarly.
For such higher order commutators, 
a natural assumption to replace \cref{locassum} is that
\begin{align}
\sum_{X\ni x} \norm{h_X} |X|^\beta \exp(\mu \,\diam(X)) \leq \zeta < \infty
\end{align}
where $\beta$ is the order of the commutator that we bound 
(i.e, $\beta=2$ for \cref{commcond}, $\beta = 3$ for \cref{highercomm}, and so on).

\subsection{Strictly Local Hamiltonians}
\label{sl}
In this subsection, we consider strictly local interactions.
We assume that
\begin{align}
H=\sum_X h_X,
\end{align}
where $h_X$ is supported on set $X$ which we assume obeys
\begin{align}
\label{metricbound}
\diam(X)\leq 1.
\end{align}
Note that any bounded range interaction (for example, with terms supported on sets of diameter at most $R$ for some given $R$) 
can be written in this form by re-scaling the metric by a constant. 
Further, we assume a bound on the commutator
\begin{align}
\label{cbsl}
\Vert [h_X,h_Y] \Vert \leq 2K
\end{align}
for some constant $K$.

We make no assumption on the bound on the norm of $h_X$.  
However, we emphasize that this does not mean that we are considering unbounded operators.  
Rather, we assume that all $h_X$ have bounded norm 
(which, in conjunction with an assumption of finite system size, 
means that $H$ has bounded norm allowing us to use various analytic estimates similar to those above),
but all bounds will be uniform in the bound on the norm of $h_X$ as well as the system size.

Define $C_B(X,t)$ and $D_B(X,t)$ as above.
Let $\IS$ denote the collection of sets $Z$ such that $\diam(Z) \leq 1$.
We use \cref{int1} from before:
\begin{align}
\label{int1p}
C_B(X,t) \leq C_B(X,0)+\sum_{Z \in \IS: Z \sim X} 2 \int_0^{|t|} D_B(Z,s) \rd s,
\end{align}
where we  explicitly write the requirement that $Z\in \IS$
as well as a version of \cref{int2} that follows from \cref{cbsl}
\begin{align}
\label{int2p}
D_B(X,t) 
\leq 
D_B(X,0) +  
\sum_{Z \in \IS: Z \sim X} 2K\int_0^{|t|} C_B(Z \cup X,s) {\rm d}s.
\end{align}

We write $X \sim\sim Z$ if there exists $Y\in \IS$ such that $X\sim Y \sim Z$;
this is equivalent to requiring that $\dist(X,Z)\leq 1$.
We write $!X\sim\sim Z$ if it is not the case that $X \sim \sim Z$.
Thus,
\begin{align}
\label{recur2}
\dist(X,Y) > 1  \; \Longrightarrow \;
C_B(X,t)& \leq 4K \sum_{Z_1,Z_2 \in \IS: X \sim Z_1 \sim Z_2} \int_0^{|t|} C_B(Z_1 \cup Z_2,s) (|t|-s) \rd s \\ \nonumber
\dist(X,Y)\leq 1 \; \Longrightarrow \;
C_B(X,t)& \leq 2 \norm{B}.
\end{align}
Here the first case follows
by combining \cref{int1p,int2p} with the condition that $D_B(Z,0)=0$ for $Z\not\sim Y$.
The substitution of $D_B(Z,s)$ in \cref{int1p} by \cref{int2p} gives a double integral
but we can change the integration order to give $(|t|-s)$ factor.
The second case is trivially true for any choice of $X$.

Thus, iterating \cref{recur2} we find that for $\dist(X,Y)>1$ we have
\begin{align}
\nonumber &C_B(X,t)
\\ \nonumber
 &\leq  4K \sum_{Z_1,Z_2 \in \IS: X \sim Z_1 \sim Z_2} \int_0^{|t|} C_B(Z_1 \cup Z_2,s_1) (|t|-s_1) \rd s_1 \\ \nonumber
& \leq  8K\norm{B} \sum_{Z_1,Z_2 \in \IS: X \sim Z_1 \sim Z_2 \sim\sim Y} \int_0^{|t|} (|t|-s_1) \rd s_1
\\ \nonumber
&\quad +(4K)^2 \sum_{\begin{subarray}{l} Z_1,Z_2,Z_3,Z_4\in \IS :\\X\sim Z_1 \sim Z_2, (Z_1 \cup Z_2) \sim Z_3 \sim Z_4, \\ \\ !Z_2\sim\sim Y\end{subarray}}
\int_0^{|t|} (|t|-s_1) \int_{0}^{|s_1|} C_B(Z_3\cup Z_4,s_2) |s_1-s_2| \rd s_2 \rd s_1 \\ \nonumber
&\leq 
8K\norm{B}\sum_{Z_1,Z_2 \in \IS: X \sim Z_1 \sim Z_2 \sim\sim Y} \int_0^{|t|} (|t|-s_1) \rd s_1
\\ \nonumber
&\quad +2\cdot(4K)^2 \norm{B}\sum_{\begin{subarray}{l} Z_1,Z_2,Z_3,Z_4\in \IS :\\X\sim Z_1 \sim Z_2, (Z_1 \cup Z_2) \sim Z_3 \sim Z_4 \sim \sim Y, \\ \\ !Z_2\sim\sim Y\end{subarray}}
\int_0^{|t|}(|t|-s_1) \int_{0}^{|s_1|}  |s_1-s_2| \rd s_2 \rd s_1 \\ \nonumber
&\quad +(4K)^2 \sum_{\begin{subarray}{l} Z_1,Z_2,Z_3,Z_4,Z_5,Z_6\in \IS :\\X\sim Z_1 \sim Z_2, (Z_1 \cup Z_2) \sim Z_3 \sim Z_4, \\(Z_3 \cup Z_4) \sim Z_5\sim Z_6, \\ \\ !Z_2\sim\sim Y, !Z_4\sim\sim Y\end{subarray}}
\int_0^{|t|}(|t|-s_1) \int_{0}^{|s_1|}  |s_1-s_2| \int_0^{|s_2|} C_B(Z_5\cup Z_6,s_3)|s_3-s_2|  \rd s_3 \rd s_2 \rd s_1
 \\ \nonumber
&\le \cdots 
\end{align}
We continue recurring in this fashion, using \cref{recur2} to substitute for $C_B$.

Thus, we find that $C_B(X,t)$ is bounded by the sum, over $k>0$ and over sequences 
$X\sim Z_1 \sim Z_2, (Z_1 \cup Z_2) \sim Z_3 \sim Z_4, \cdots \sim Z_{2k} \sim \sim Y$ where $Z_j \in \IS$ such that 
for no $j<k$ do we have $Z_{2j} \sim \sim Y$, of $2\norm{B} (4K)^{k} |t|^{2k}/(2k)!$.
Since all terms in the sum are positive,
this is bounded by the sum over all sequences $X \sim Z_1  \sim \ldots \sim Z_{2k} \sim \sim Y$
of $2\norm{B}(4K)^{k} |t|^{2k}/(2k)!$,
i.e., we may remove the restriction $! Z_{2j} \sim\sim Y$.
Further, we may relax the restriction $(Z_{2j-1} \cup Z_{2j}) \sim Z_{2j+1}$ to $Z_{2j} \sim \sim Z_{2j+1}$.

For any graph with bounded degree $d$, 
this gives a Lieb-Robinson velocity $v_{LR}$ bounded by a constant times $d \sqrt{K}$.
Indeed, the sum is bounded by a sum over both even and odd length sequences
$X\sim Z_1 \sim Z_2 \sim \sim Z_3 \sim Z_4 \sim \sim Z_5 \ldots \sim Z_n \sim \sim Y$ with any $n \ge \dist(X,Y)$ of $2\norm{B} (2\sqrt{K})^n |t|^n/n!$.
This is the same sum as appears in the Lieb-Robinson bound for strictly decaying interactions,
whose strength is bounded in norm by $\sqrt{K}$.

Remark: As in the proof of the usual Lieb-Robinson bound, 
the convergence of the sum of the first $n$ terms of the sequence to $C_B(X,t)$ as $n\rightarrow \infty$ 
can be established by bounding the remainder term which is proportional 
to a sum of $C_B(Z_{2k-1}\cup Z_{2k},s_k)$, and using $C_B\leq 2 \norm{B}$.
In fact, 
since the sum over sequences $X \sim Z_1 \sim \cdots \sim Z_{2k} \sim \sim Y$ is empty if $3k+1 < \dist(X,Y)$,
the remainder for the sum of first $\ell=\lfloor \dist(X,Y) \rfloor$ terms is 
so small that we can conclude $v_{LR} = \calO(d\sqrt K)$ without dealing with an infinite series.
This was the proof method in \cref{app:pfLRB}.
One disadvantage of this simpler argument is that the resulting bound on the Lieb-Robinson velocity
is slightly worse than what is obtainable by the infinite series, 
particularly if the interaction graph is expanding.

\section{Analytic functions are efficiently computable}
\label{app:chebyshev}

Here we briefly review a polynomial approximation scheme~\cite{ATP}
to {\it analytic} functions
--- 
that have power series representations.%
\footnote{
An example of smooth but non-analytic function is $\exp(-z^{-2})$.
This function fails to be analytic at $z=0$.
It is infinitely differentiable at $z=0$,
with derivatives all zero, and hence the Taylor series is identically zero,
but the function is not identically zero around $z=0$.
If a real function is analytic at $x \in \RR$,
then its power series converges in an open neighborhood of $x$ in the complex plane
(analytic continuation).
}
For $\rho > 1$, define $E_\rho$ to be the {\it Bernstein ellipse},
the image of a circle $\{ z  \in \CC: |z| = \rho \}$
under the map $z \mapsto \frac12( z + z^{-1})$.
The Bernstein ellipse always encloses $[-1,1]$,
and collapses to $[-1,1] \subset \CC$ as $\rho \to 1$.
It is useful to introduce {\it Chebyshev polynomials} (of the first kind)
$T_m$ of degree $m \ge 0$, defined by the equation
\begin{align}
	T_m( \frac{z+z^{-1}}{2} ) = \frac{z^m + z^{-m}}{2}.
	\label{eq:chebyshev-def}
\end{align} 
Picking $z$ on the unit circle, we see that 
$T_m : [-1,1] \ni \cos \theta \mapsto \cos m\theta \in [-1,1]$.

\begin{lemma}\label{eq:chebyshev-approximation}
Let $f$ be a analytic function on the interior of $E_\rho$ for some $\rho>1$,
and assume $\sup_{z \in E_\rho} f(z) = M < \infty$.
Then, $f$ admits an approximate polynomial expansion in Chebyshev polynomials
such that
\begin{align}
\max_{x \in [-1,1]} 
\abs{f(x) - \sum_{j=0}^J a_j T_j(x) } \le \frac{2M}{\rho - 1} \rho^{-J} .
\end{align}
\end{lemma}
\begin{proof}
The series expansion is a disguised cosine series:
\begin{align}
	f(\cos \theta) 
	= \sum_{j=0}^\infty a_j T_j(\cos \theta) 
	= \sum_{j=0}^\infty a_j \cos( j \theta).
\end{align}
It is important that $\theta \mapsto f(\cos\theta)$ is a periodic smooth function,
and therefore its Fourier series converges to the function value.
The coefficients can be read off by
\begin{align}
	a_j &= \frac{1}{\pi} \int_0^{2\pi}\rd \theta f(\cos\theta) \cos j \theta & (j > 0) ,\\
	a_0 &= \frac{1}{2\pi} \int_0^{2\pi} \rd \theta f(\cos\theta)  &(j=0).
\end{align}
The assumption on the analyticity
means that $z \mapsto f(\frac 12 (z+z^{-1}))$ is analytic
on the annulus $\{ z \in \CC : \rho^{-1} < |z| < \rho \}$.
Thus, we can write the above integral as a contour integral along the circle
$C_\rho = \{ z \in \CC : |z| = \rho\}$,
and obtain a bound for $j > 0$
\begin{align}
|a_j |
= 
\abs{
\frac{1}{\i \pi} \int_{C_\rho} \frac{\rd z}{z} 
f\left(\frac{z+z^{-1}}{2}\right)\frac{z^j + z^{-j}}{2}
}
=
\frac{1}{\pi} \abs{
	\int_{C_\rho} \frac{\rd z}{z} ~f\left(\frac{z+z^{-1}}{2}\right) z^{-j}
} 
\le 
\frac{2M}{\rho^{j}}
\end{align}
where the second equality is due to the symmetry $z \leftrightarrow z^{-1}$ of the integrand.
For $j=0$, we know $\abs{a_0} \le M$.
Since $\abs{T_j(x)} \le 1$ for $x \in [-1,1]$, this completes the proof.
\end{proof}	
Therefore, a function on $[-1,1]$ that is analytic over $E_\rho$
can be computed to accuracy $\epsilon$ 
by evaluating a polynomial of degree $\calO(\log(1/\epsilon))$.

\section{Hamiltonian simulation by quantum signal processing and Qubitization}
\label{app:qspalg}

In this section, we outline Hamiltonian simulation of $e^{-itH}$ 
by quantum signal processing~\cite{QSP} and Qubitization~\cite{LC16} in three steps, 
and introduce a simple situational circuit optimization for constant factor improvements in gate costs. 

First, one assumes that the Hamiltonian $H$ acting on register $s$ is encoded in a certain standard-form:
\begin{align}
(\bra{G}_a\otimes \id_s)O(\ket{G}_a\otimes \id_s) = H/\alpha.
\label{eq:standard_form}
\end{align} 
Here, we assume access to a unitary oracle $O$ that acts jointly on registers $a,s$, and a unitary oracle $G$ that prepares some state $G\ket{0}_a=\ket{G}_a$ such that~\cref{eq:standard_form} is satisfied with some normalization constant $\alpha \ge \|H\|$. Note that $\alpha$ represents the quality of the encoding; a smaller $\alpha$ leads fewer overall queries to $O$ and $G$.

Second, the Qubitization algorithm queries $O$ and $G$ 
to construct a unitary $W$, the qubiterate, with eigenphases $\theta_\lambda=\sin^{-1}{(\lambda/\alpha)}$ directly related to eigenvalues of the Hamiltonian $H\ket{\lambda}=\lambda\ket{\lambda}$. 
In the case where $O^2=\id_{as}$, 
this is accomplished by a reflection about $\ket{G}_a$:
\begin{align}
W &= 
-\i\left((2\ket{G}\bra{G}_a -\id_a)\otimes\id_s\right) O.
\label{eq:qubitization}
\end{align} 
For every eigenstate $\ket{\lambda}$, the normalized states
\begin{align}
&\ket{G_\lambda}=\ket{G}_a\ket{\lambda}_s,
&
 \ket{G_\lambda^\perp}\propto (1-\bra{G_\lambda}W\ket{G_\lambda})\ket{G_\lambda},
 \\\nonumber
&\ket{G_{\lambda\pm}}=\frac{1}{\sqrt{2}}(\ket{G_\lambda}\pm \i\ket{G_\lambda^\perp})
&
\end{align}
are eigenstate of $W$ with eigenvalues 
\begin{align}
W\ket{G_{\lambda\pm}}=\mp e^{\pm \i \theta_\lambda}\ket{G_{\lambda\pm}},\quad
\theta_\lambda = \sin^{-1}{(\lambda/\alpha)}.
\end{align}

Third, the quantum signal processing algorithm queries the qubiterate 
to approximate a unitary $V$ which has the same eigenstates $\ket{G_{\lambda\pm}}$, 
but with eigenphases transformed as 
\begin{align}
\mp e^{\pm \i \theta_\lambda} \mapsto e^{-\i \alpha t\sin{\theta_\lambda}} = e^{-\i t \lambda}.
\end{align}
That is, $V$ always has both $\ket{G_\lambda}$ and $\ket{G_\lambda^\perp}$ in an eigenspace.
Therefore, the time evolution by $e^{-itH}$ is accomplished as follows:
\begin{align}
& 
V\ket{G}_a \ket{\lambda}_s
= V\frac{\ket{G_{\lambda+}}+\ket{G_{\lambda-}}}{\sqrt{2}}
= e^{-\i t \lambda} \ket{G}_a \ket{\lambda}_s
\\ \nonumber
&\Longrightarrow (\bra{G}_a\otimes \id_s)V(\ket{G}_a\otimes \id_s) = e^{-itH}.
\end{align}
The transformation from $W$ to $V$
is accomplished by a unitary sequence $V_{\vec\varphi}$
such that $\bra{+}_b (V_{\vec \varphi})_{bas} \ket{+}_b \simeq V$
where $b$ is a single-qubit ancilla, and $\ket \pm =\pm X \ket \pm$.
$V_{\vec \varphi}$ is a product of controlled-$W$ 
interspersed by single-qubit rotations,
parameterized by $\vec\varphi \in\mathbb R^N$ with $N$ even,
defined as follows:
\begin{align}
V_{\vec\varphi}
&=
V^\dag_{\pi+\varphi_N}V_{\varphi_{N-1}}
V^\dag_{\pi+\varphi_{N-2}}V_{\varphi_{N-3}}
\cdots
V^\dag_{\pi+\varphi_2}V_{\varphi_{1}},
\\\nonumber
V_{\varphi}
&=
(e^{-\i\varphi Z/2} \otimes \id_{as})
(\ket{+}\bra{+}_b\otimes \id_{as} + \ket{-}\bra{-}_b\otimes W)
(e^{\i\varphi Z/2}\otimes \id_{as}).
\end{align}
For every eigenstate $W\ket{\theta}_{as}=e^{\i\theta}\ket{\theta}_{as}$, 
the operator $V_{\varphi}$ on a given $\ket{\theta}_{as}$ 
introduces a phase kickback to the ancilla register $b$.
The net action on the ancilla $b$ given $\ket{\theta}_{as}$ is
\begin{align}
e^{-\i\varphi Z/2} \left( e^{\i\theta / 2} e^{-\i\theta X / 2} \right) e^{\i \varphi Z / 2}
=
e^{\i\theta / 2}e^{-\i \theta P_\varphi / 2}, \quad P_\varphi= X \cos \varphi + Y \sin \varphi.
\end{align}
Thus by multiplying out the single qubit-rotations,
\begin{align}
(\bra{+}_b \bra{\theta}_{as})V_{\vec\varphi}(\ket{+}_b \ket{\theta}_{as})
&= 
\bra{+}_b 
e^{-\i \theta P_{\varphi_N}/2}
e^{-\i \theta P_{\varphi_{N-1}}/2} 
\cdots 
e^{-\i \theta P_{\varphi_1}/2}
\ket{+}_b
\nonumber\\
&=
\bra{+}_b \left( 
\sum^{N/2}_{k=0}(c^\id_k \id + \i c^Z_k Z)\cos k\theta + \i(c^X_k X + c^Y_k Y) \sin k \theta
\right) \ket{+}_b,
\nonumber\\
&=
\sum^{N/2}_{k=0}c^\id_k\cos{k\theta} + \i c^X_k \sin{k\theta},
\label{eq:phaseTransf}
\end{align}
where $\{c^\id, c^X, c^Y, c^Z\}$ are real coefficients determined by $\vec \varphi$.
(For the second equality, it is useful to work out $N=2$ case.)
We would like \cref{eq:phaseTransf} to be $e^{-\i t \sin \theta}$,
or a good approximation thereof. 
The following Jacobi-Anger expansion suits this purpose.
\begin{align}
e^{-\i \alpha t \sin{\theta}}
=\left(J_0(\alpha t)
+2\sum^\infty_{k: \text{ even } >0} J_{k}(\alpha t)\cos{k\theta}\right)
-\i
\left(2\sum^\infty_{k: \text{ odd }>0}J_{k}(\alpha t)\sin{k\theta}\right)
\label{eq:JacobiAnger}
\end{align}
where $J_k$ is the Bessel function of the first kind.
Remark that the function $e^{-\i\theta} \mapsto e^{-\i\alpha t \sin \theta}$
sends both $e^{-\i\theta}$ and $- e^{\i\theta} = e^{\i(\theta +\pi)}$ to the same value,
and the same property holds in the right-hand side of \cref{eq:JacobiAnger} term-by-term.
If we keep the series up to order $N/2$,
the truncation error of this approximation is at most $2 \sum_{k > N/2} |J_{k}(\alpha t)|$.
In principle, the angles $\vec{\varphi}$ that generate the desired coefficients $\{c^\id,c^X\}$
may be precomputed by a classical polynomial-time algorithm~\cite{LowMethod,Haah2019product}
given $\alpha \ge \norm{H}$ and $t$.
This ultimately leads to an approximation of $e^{-\i \alpha t \sin{\theta}}$ 
with error
\begin{align}
\epsilon_\square=
\norm{
(\bra{+}_b \bra{G}_{a}\otimes \id_s)V_{\vec\varphi}(\ket{+}_b \ket{G}_{a}\otimes \id_s)-e^{-\i tH}
}
\le 
16\sum^\infty_{k=q}|J_{k}(\alpha t)|\le \frac{32(\alpha t)^q}{2^q q!}
,
\quad 
q = \frac{N}{2}+1.
\end{align}
(The extra factor of 8 in the estimate is mainly because
we cannot guarantee that the truncated series 
is exactly implemented by unitary $V_{\vec \varphi}$.)
When evaluating gate counts of this simulation algorithm, 
we will use placeholder values for $\vec{\varphi}$.
The algorithm requires a post-selection;
however, the success probability is $1-\calO(\epsilon_\square)$
since the post-selected operator is $\epsilon_\square$-close to a unitary.

\subsection{Encoding coefficients in reflections}

With these three steps, the remaining task is 
to construct the oracles $O$ and $G$ 
that encode the desired Hamiltonian. 
For a general Hamiltonian represented as a linear combination of $M$ Pauli operators $P_j$
\begin{align}
H=\sum_{j=0}^{M-1}\alpha_j P_j,\quad \alpha_j>0, \quad
\alpha=\sum_{j=0}^{M-1}\alpha_j,
\end{align}
the most straightforward approach defines
\begin{align}
O=\sum_{j=0}^{M-1}\ket{j}\bra{j}\otimes P_j,
\quad
\ket{G}_a=\sum_{j=0}^{M-1}\sqrt{\frac{\alpha_j}{\alpha}}\ket{j}_a.
\end{align}
It is easy to verify that 
$(\bra{G}_a\otimes \id_s)O(\ket{G}_a\otimes \id_s) = H/\alpha$, and $O^2=\id_{as}$, as desired. 
The gate complexity $O$ and $G$ are asymptotically similar: 
The control logic for $O$ may be constructed using $\calO(M)$ 
NOT, CNOT, and Toffoli gates~\cite{Childs2017}, 
and the creation of an arbitrary dimension $M$ quantum state requires 
$\calO(M)$ CNOT gates and arbitrary single-qubit rotations~\cite{shende2006synthesis}.
Each single-qubit rotation may then by approximated to error $\epsilon$ by a sequence of $\mathcal{O}(\log{1/\epsilon})$ Clifford+$T$ gates~\cite{Kliuchnikov2013Synthesis} -- if the overall algorithm uses $N$ single-qubit rotations, a triangle inequality bounds the total error to at most $N\epsilon$.

However, if many coefficients $\alpha_j$ of H are identical,
the use of arbitrary state preparation is excessively costly.
For instance, in the extreme case where all $\alpha_j$ are identical,
and $M$ is a power of $2$, $\log_2{M}$ Hadamard gates suffice to prepare $\ket{G}_a$ 
--- when $M$ is not a power of two, 
a uniform superposition over all states up to $M$ 
may still be prepared with cost $\calO(\log{M})$ 
by combining integer arithmetic with amplitude amplification.
Or, we can add extra Hamiltonian terms $\id$ to make the number of terms to be a power of 2;
this shifts energy level of the Hamiltonian and 
gives a time-dependent global phase factor to the time evolution unitary.
In another case where only a few $\alpha_j$ differ,
one may exploit a unary representation of the control logic~\cite{poulin2017fast} 
to accelerate the preparation of $\ket{G}$ at the cost of additional space overhead.

Rather than encoding coefficient information in the state $\ket{G}$,
one simple alternate approach is to encode coefficient information 
by replacing the operators $P_j$ with exponentials $U_j=e^{-\i P_j \cos^{-1}(\alpha_j)}$ 
by taking a linear combination of $\frac{U_j+U^\dag_j}{2}=\alpha_j P_j$.
However, as $U^2_j\neq\id$ in general, 
the downside of this approach is that $O^2\neq\id_{as}$,
which violates the prerequisite for the simple qubitization circuit of~\cref{eq:qubitization}.

We present a simple modification that allows us to encode coefficient information 
in unitary operators whilst maintaining the condition $O^2=\id_{as}$.
Consider the two-qubit circuit acting on register $c$.
\begin{align}
Q_j=(\id \otimes e^{\i \beta_j X})
\operatorname{SWAP}
(\id \otimes e^{-\i \beta_j X}), 
\quad
Q_j^2 = \id_{12}
,
\quad
\bra{00}Q_j\ket{00}= \cos^2 \beta_j
\end{align}
where $\beta_j>0$. Thus if we define
\begin{align}
O=\sum_{j=0}^{M-1}\ket{j}\bra{j} \otimes Q_j\otimes P_j,
\quad
\ket{G}_{ac} = \sum_{j=0}^{M-1} \frac{1}{\sqrt{M}} \ket{j}_a \ket{00}_c,
\end{align}
then $O^2=\id_{abc}$ and this encodes the Hamiltonian
\begin{align}
(\bra{G}_{ac}\otimes \id_s) O (\ket{G}_{ac}\otimes \id_s) 
= \frac{1}{M}\sum_{j=0}^{M-1}P_j \cos^2{\beta_j}.
\end{align}
This construction can be advantageous in the situation, 
such as in~\cref{eq:Hamiltonian} where most coefficients are $1$, 
and only a few are less than one 
--- whenever $\beta_j=0$, we replace $Q_j$ with the identity operator.

\bibliographystyle{alphaurl}
\bibliography{lhs-ref}
\end{document}